\documentclass{article}
\newif\ifarxiv
\arxivtrue

\usepackage{ijcai24}

\usepackage{times}
\usepackage{soul}
\usepackage{url}
\usepackage[hidelinks]{hyperref}
\usepackage[utf8]{inputenc}
\usepackage[small]{caption}
\usepackage{graphicx}
\usepackage{booktabs}
\usepackage[switch]{lineno}

\usepackage{macros_pan_23}
\usepackage{macros_yohai}

\title{Design a Win-Win Strategy That Is Fair to Both Service Providers and Tasks When Rejection Is Not an Option}
       
\author{
Yohai Trabelsi$^1$
\and
Pan Xu$^2$\and
Sarit Kraus$^{1}$\\
\affiliations
$^1$Department of Computer Science, 
Bar-Ilan University, Ramat Gan, Israel\\
$^2$New Jersey Institute of Technology, Newark, NJ, USA\\
\emails
yohai.trabelsi@gmail.com,
pxu@njit.edu,
sarit@cs.biu.ac.il
}

\begin{document}
\maketitle 

\begin{abstract}

Assigning tasks to service providers is a frequent procedure across various applications.  Often the tasks arrive dynamically while the service providers remain static. Preventing task rejection caused by service provider overload is of utmost significance.
To ensure a positive experience in relevant applications for both service providers and tasks, fairness must be considered. 
To address the issue, we model the problem as an online matching within a bipartite graph and tackle two minimax problems: one focuses on minimizing the highest waiting time of a task, while the other aims to minimize the highest workload of a service provider. We show that the second problem can be expressed as a linear program and thus solved efficiently while maintaining a reasonable approximation to the objective of the first problem. We developed novel methods that utilize the two minimax problems. We conducted extensive simulation experiments using real data and demonstrated that our novel heuristics, based on the linear program, performed remarkably well.
\end{abstract}

\section{Introduction}
\label{sec:intro}
In resource allocation, numerous problems can be represented as online matching in bipartite graphs. One side of the graph comprises service providers (interchangeably called workers in this paper), while the other consists of allocated task types. The graph's edges indicate the qualifications of service providers to perform tasks of specific types.

In online matching problems, a common scenario involves one dynamic side and one static side. This dynamic-static setup finds application in various contexts, such as matching riders(dynamic) to drivers(static) \cite{dickerson2021allocation}, connecting search queries(dynamic) to advertisers in sponsored search(static) \cite{delong2022online}, and facilitating the teleoperation of autonomous vehicles (AVs) \cite{viden}.
The primary objective in these problems is to optimize some criteria from the perspective of the allocator. 

Some other works are dedicated to optimizing allocation fairness. For example, in the domain of ride-sourcing, a method to achieve allocation fairness was proposed in 
\cite{lesmana2019balancing}. Additionally, certain studies address cases where fairness should be maintained for both online tasks and offline workers \cite{esmaeili2022rawlsian}.

Our work is motivated by the teleoperation of AVs that has garnered increasing attention recently ~(e.g., \cite{zhang2020toward,viden,tener2022driving}). The primary role of teleoperation is to aid AVs by intervening in challenging driving situations\footnote{As mentioned in~\cite{tener2022driving}, the AVs will need this intervention, at least in the near future.}. Ensuring a fair allocation of teleoperators to driving tasks is crucial for enhancing the satisfaction of both teleoperators and AVs' users. Particularly, if certain intervention requests have significantly longer waiting times or if some teleoperators are disproportionately busier than others, such imbalances can lead to dissatisfaction among those affected. In addition, as a person in the vehicle is awaiting the teleoperator's intervention, a rejection of a request is unacceptable. Another property of this application is that the teleoperators (workers) are reusable, which means they are ready to perform a new intervention request (task) once they finish a previously allocated request.

We model the problem as online matching in a bipartite graph and propose several approaches to optimize fairness for both the tasks (e.g., intervention requests) and the workers (e.g., teleoperators) involved in the process. Our notion of fairness is aligned with Rawls' theory of justice \cite{rawls1999theory}.

We introduce two minimax problems within the given context. The first concerns fairness regarding tasks relative to waiting times, while the second focuses on Rawlsian fairness for service providers based on their workload. In both scenarios, task rejection is not permissible. We demonstrate that the second problem can be efficiently formulated as a linear problem. Notably, the solution to the second problem mirrors the first when task durations from each worker conform to the same distribution. In cases where this isn't true, we show that the second problem's solution approximates the first problem's solution, supported by a provable approximation ratio.
Our study concludes with extensive simulations that underscore the efficacy of these minimax problems. Furthermore, we devise innovative heuristics that leverage the minimax solutions. These heuristics enhance task fairness while preserving favorable outcomes for worker fairness.

{\bf Our main contributions are:}
(1) We propose two models to promote fairness among tasks and workers. 
(2) We present an LP-based algorithmic framework, which can exactly solve fairness maximization among workers and approximately among tasks, and we provide a tight approximation bound. 
(3) We empirically implement and compare different methods, including several baselines, on datasets involving the teleoperation of AVs.

\subsection{Related Work}

In this section, we describe previous works about fair allocation and allocation with delays. Notably, to our knowledge, our work distinguishes itself by being the first to consider fairness and allocation delays together.

\paragraph{Fair allocation}

Some studies address fair allocation, focusing on only one side of the graph, as seen in \cite{ma2020group}. Although their fairness approach resembles ours, it pertains solely to one side of the graph, which falls short of our requirements. Other research, like \cite{patro2020fairrec}, deals with fairness in recommendation systems. However, the fairness objectives in recommendation systems significantly differ from those in task allocation contexts. Practical solutions for enhancing fairness for both service providers and tasks are explored in works such as \cite{zhou2023subgroup}. Regrettably, this branch of research lacks theoretical performance bounds for their solutions.
The fairness principles in~\cite{esmaeili2022rawlsian} closely align with ours. They consider both workers (offline side) and tasks (online side), embracing Rawlsian welfare \cite{rawls1958justice}. Nonetheless, task rejection is permissible in their scenario if workers are unavailable.

\paragraph{Allocation with delayed assignments}
The original online matching problem was introduced in \cite{karp}, where static nodes (workers) are instantly paired with dynamic nodes (tasks) upon arrival. However, real scenarios often lack immediate worker availability for tasks, prompting consideration for task execution delays over outright rejection.
Numerous works tackle resource allocation with potential task delays. However, many of these approaches (e.g., \cite{righter1987stochastic,li2023fully}) prioritize utility maximization without factoring in task wait times or worker workload. Some leverage reinforcement learning for such issues yet often make batch decisions, leading to suboptimal outcomes. Moreover, theoretical guarantees are frequently absent.
An LP-based method for delayed allocations is presented in \cite{viden}, optimizing a complex utility function that accounts for task waiting times but overlooks worker workload.

Another pertinent domain involves queue admission control systems with multiple classes. Here, diverse customer types (tasks) arrive dynamically, and a decision-maker determines which task to accept, as demonstrated in \cite{rigter2022optimal}. However, several studies in this realm do not distinguish between workers, while others permit task rejection.
To our knowledge, the problem of two-sided fair allocation when task rejection is not allowed has not yet been addressed.

\section{Preliminaries}
\setlength{\tabcolsep}{2pt}
\begin{table}[t!]
\label{table:notation}
\begin{tabular}{ll } 
 \hline
$G$ & Input network graph $G=(I,J,E)$. \\
$I$ ($J$) & Set of worker (task) types.\\
$\cN_i$ ($\cN_j$) & Set of neighbors of $i$ ($j$).  \\
$i \sim j$ ($j \sim i$) & Equivalent to $i \in \cN_j$ ($j \in \cN_i$). \\
$\lam_j$ & Arrival rate of task type $j \in J$.\\
$\lam_i$ & Arrival rate on worker $i \in I$.\\
$\Exp(\mu)$ & Exponential distribution of rate $\mu>0$.\\
 $\Exp(\mu_{ij})$ & Service time taken by worker $i$ to service $j$.\\
$\rho_i \in [0,1]$ & Workload of worker $i \in I$.\\
$w_j$ & Expected (absolute) waiting time of $j$; see Eqn.~\eqref{eqn:w_j}. \\
$\bw_j$ & Expected (relative) waiting time of $j$; see Eqn.~\eqref{eqn:bw_j}. \\
$\kap \ge 1$ & $\max_{i \in I} \sbp{\max_{j \sim i,j'\sim i}\mu_{ij}/\mu_{i,j'}}$.\\
 \hline
\end{tabular}
\small
\caption{A glossary of notations  throughout this paper.}
\end{table}
Suppose we use a bipartite graph $G=(I,J,E)$ to model the worker-task network, where $I$ denotes the set of offline workers (\eg teleoperators), $J$ the set of types of tasks, and an edge $e=(i,j)$ indicates the feasibility of worker $i$ to serve the task (of type) $j$. 
Note that at certain points within this paper, we abuse the notation by referring to $j$ as a task instead of a task type. We also abuse the notation by referring to an edge $e=(i,j)$ as $(ij)$.
Tasks of type $j \in J$ arrive following an independent Poisson process of rate $\lam_j>0$. For each edge $e=(i,j) \in E$, we assume it takes worker $i$ an exponentially distributed service time\footnote{This assumption is justified in \cite{devore2008probability}. Note that the theoretical analysis does not depend on it. We could use any distribution if the mean and the variance of service time are known.} of rate $\mu_{ij}>0$ to complete a task of type $j$ (\ie with mean of $1/\mu_{ij}$)\footnote{Note that the assumption does not necessarily suggest the most likely outcome is for tasks to be finished in an extremely short time. 
Consider a task type with an exponentially distributed service time of rate $\mu$, denoted as $X=\Exp(\mu)$. We observe that for any given threshold $a>0$, $\Pr[X \ge a] =e^{-\mu a}$, which can be close to one when $\mu$ is small.
}.  For each worker $i$ and task $j$, let $\cN_i \subseteq J$ and $\cN_j \subseteq I$ denote the set of neighbors of $i$ and $j$ in the graph $G$. \tbf{The assigning rule} is as follows. Upon the arrival of a task of type $j$, we (as the central coordinator)  have to assign it to a feasible worker $i \in \cN_j$ immediately: if $i$ is free (or available) at that time, then $i$ will serve $j$ right away; otherwise, $j$ will join the virtual queue of $i$ and it will stay there until being served by $i$.  

\subsection{Allocation Policy and Related Concepts}
Consider an allocation policy $\pi(\x)$ (possibly randomized), characterized as a vector $\x=\{x_{ij}| (ij) \in E\}$, where $x_{ij} \in [0,1]$ denotes the percentage of task (of type) $j$ assigned to and served by worker $i$. In the following, we discuss a few important properties and concepts related to $\pi(\x)$. Let $\cQ_i$ be the virtual queue maintained by worker $i \in I$.

\xhdr{Arrival rate on $\cQ_i$, denoted by $\lam_i$}. Observe that $\x=(x_{ij})$ can be viewed alternatively as the probability that $\pi$ assigns each arriving $j$ to $i$. Thus, we claim that $\cQ_i$ admits a Poisson arrival process of rate $\lam_i:=\sum_{j \in \cN_i} \lam_j \cdot x_{ij}$.
By the property of the Poisson process (See section 2.3.2 at \cite{gallager2011discrete}), conditioning on the arrival of task (of type) 
$\bj \in J$ on $i$,
we claim that $\Pr[\bj=j]=x_{ij} \cdot \lam_j/\lam_i$ for each $j \in \cN_i$.

\xhdr{Service time on $\cQ_i$, denoted by $\cS_i$}. The analysis above shows that the task joining $\cQ_i$ is of type $j \in \cN_i$ with probability equal to  $ x_{ij} \cdot \lam_j/\lam_i$. Thus, the overall service time $\cS_i=\sum_{j \in \cN_i} \chi_{ij} \cdot \Exp(\mu_{ij})$, where $\chi_{ij}=1$ indicates that the task joining $i$ is of type $j$ with $\E[\chi_{ij}]=x_{ij} \cdot \lam_j/\lam_i$, and $\Exp(\mu_{ij})$ represents the exponentially distributed service time of $i$ for $j$ of rate $\mu_{ij}$. Thus, $\cS_i$ follows a \emph{hyperexponential distribution}~\cite{gupta1964queues} with mean equal to 
\begin{align}\label{eqn:s_i}
s_i:=\E[\cS_i]=\sum_{j \in \cN_i} (x_{ij}  \lam_j)/(\lam_i  \mu_{ij}).
\end{align}

\xhdr{Workload of worker $i$, denoted by $\rho_i$}. By definition, 
\begin{align}\label{eqn:rho_i}
\rho_i:=\lam_i \cdot \E[\cS_i]=\sum_{j \in \cN_i} (x_{ij}  \lam_j)/ \mu_{ij},
\end{align}
where $\rho_i$ can be re-interpreted as the probability that the worker $i$ is busy or the
proportion of time the worker $i$ is busy averaged over a long period. Note that $\rho_i<1$ is the key condition ensuring the virtual queue $\cQ_i$ can enter a stable state. This is also a condition we should impose on every worker $i \in I$ when designing policy $\pi(\x)$ since otherwise, $i$ could always stay occupied in the long run (thus, not acceptable to $i$) and every task $j$ assigned to $i$ could risk an infinitely long waiting time (not acceptable to $j$).

\xhdr{Waiting time on worker $i$, denoted by $W_i$}. By the analysis above, we see that the queue $\cQ_i$ on worker $i$ qualifies as an $M/G/1$ (using the standard Kendall’s notation \cite{kendall1953stochastic}), which means it admits a Poisson arrival process, a general service time distribution, and a single worker. By the Pollaczek-Khinchin mean formula~\cite{asmussen2003random}, 
\begin{align}\label{eqn:w_i}
w_i:=\E[W_i]=\frac{\lam_i \E[\cS^2_i]}{2(1-\rho_i)}=\frac{\sum_{j \in \cN_i} x_{ij} \lam_j/\mu^2_{ij}}{1-\sum_{j \in \cN_i} x_{ij} \lam_j/\mu_{ij}},
\end{align}
where the numerator is equal to
\begin{align*}
&\lam_i \cdot \E[\cS_i^2]=\lam_i \cdot  \E\bb{\bp{\sum_{j \in \cN_i} \chi_{ij}\cdot\Exp(\mu_{ij})}^2}\\
&=\lam_i \cdot \E\bb{\sum_{j \in \cN_i} \chi_{ij} \cdot\Exp^2(\mu_{ij})  }
=\lam_i \cdot \sum_{j \in \cN_i} \E \bb{ \chi_{ij} \cdot\Exp^2(\mu_{ij})}\\
&= \lam_i \cdot \sum_{j \in \cN_i}  (x_{ij} \lam_j/\lam_i) \cdot (2/\mu^2_{ij})
= 2\sum_{j \in \cN_i}  (x_{ij} \lam_j/\mu^2_{ij}).
\end{align*}

\xhdr{Absolute and relative waiting time of $j$, denoted by $W_j$ and $\obar{W}_j$}. Recall that under $\pi(\x)$, a task $j$ will be assigned to a feasible worker $i \in \cN_j$ with probability $x_{ij}$. Thus, the expected (absolute) waiting time of $j$ should be
\begin{align}\label{eqn:w_j}
w_j:=\E[W_j]=\sum_{i \in \cN_j} x_{ij} \cdot w_i,
\end{align}
where $w_i$ is the expected waiting time on queue $\cQ_i$, as shown in~\eqref{eqn:w_i}. The \emph{relative} waiting time of $j$ on $\cQ_i$ is defined as the ratio of waiting time on $\cQ_i$ to the service time of $i$ for $j$, which has a mean of $1/\mu_{ij}$. Thus, the expected relative waiting time of $j$ should be
\begin{align}
\bw_j&:=\E[\bW_j]=\sum_{i \in \cN_j} x_{ij} \cdot w_i/(1/\mu_{ij}) \nonumber\\
&=\sum_{i \in \cN_j} x_{ij} \cdot \mu_{ij} \cdot \frac{\sum_{\ell \sim i} x_{i \ell} \lam_\ell/\mu_{i \ell}^2}{1-\sum_{\ell \sim i} x_{i \ell} \lam_\ell/\mu_{i \ell}}
\label{eqn:bw_j}.
\end{align}

\subsection{Two Fairness-Related Objectives}
In this paper, we propose the following two fairness metrics and objectives when optimizing a policy $\pi(\x)$.

\paragraph{\obja: Fairness promotion among tasks, denoted by ${\min \max_{j \in J} \bw_j}$}. 
We quantify the overall fairness among users achieved by policy $\pi(\x)$ as the maximum expected \emph{relative} waiting time among all task types, \ie $\max_{j \in J} \bw_j$.  A formula for calculating the relative waiting time is shown in~\eqref{eqn:bw_j}.
Note that here we choose the relative version instead of the absolute one (\ie $\max_{j \in J} w_j$) following, for example, the  paper~\cite{maister1984psychology} that asserts that ``the more valuable the service, the longer the customer will wait." A compelling example is that:``Special checkout counters were originally provided because customers with only a few items felt resentful at having to wait a long time for what was seen as a simple transaction. Customers with a full cart of groceries were much more inclined to tolerate lines."


\xhdr{\objb: Fairness promotion among workers, denoted by $\min \max_{i \in I} \rho_i$}. Recall that for each worker $i \in I$, the workload $\rho_i \in (0,1)$, as defined in~\eqref{eqn:rho_i}, captures the percentage of busy time on worker $i$. Thus, the maximum workload, \,  i.e., $\max_{i \in I} \rho_i$,  reflects the highest degree of being occupied among all workers under policy $\pi(\x)$. By opting for minimization of the maximum workload, denoted by 
$\min \max_{i \in I} \rho_i$, we aim to minimize the occupation time of the most occupied worker as substantially as feasible.

\subsection{Two Optimization Programs}
Consider an allocation policy $\pi(\x)$ parameterized by $\x=(x_{ij})$, where $x_{ij}$ with $(ij) \in E$ denotes the percentage of task of type $j$ assigned to worker $i$. For ease of notation, we will use $i \sim j$ (and $j \sim i$) to represent $i \in \cN_j$ (and $j \in \cN_i$)   throughout this paper. We formulate \obja and \objb as minmax programs as follows. 
\begin{alignat}{2}
(\pra)\min &\max_{j \in J} \bp{\bw_j=\sum_{i \sim j} x_{ij} \cdot \mu_{ij} \cdot \frac{\sum_{\ell \sim i} x_{i \ell} \lam_\ell/\mu_{i \ell}^2}{1-\sum_{\ell \sim i} x_{i \ell} \lam_\ell/\mu_{i \ell}}},  \label{obj_a} \\
 & x_j:=\sum_{i \sim j} x_{ij} =1,  ~~\forall j \in J \label{cons:j} \\ 
 &  \rho_i=  \sum_{\ell \sim i} x_{i \ell} \lam_\ell/\mu_{i\ell} \le 1, ~~\forall i \in I \label{cons:i} \\ 
  &0 \le  x_{ij} \le 1,  ~~\forall (ij)\in E. \label{cons:e}
\end{alignat}
\begingroup
\allowdisplaybreaks
\begin{alignat}{2}
(\prb)~~\min &~~ \max_i \rho_i,&&  \label{obj_b} \\
 & x_j:=\sum_{i \sim j} x_{ij} =1,  &&~~\forall j \in J \label{cons:jb} \\ 
 &  \rho_i=  \sum_{\ell \sim i} x_{i \ell} \lam_\ell/\mu_{i\ell} \le 1,&&~~ \forall i \in I \label{cons:ib} \\ 
  &0 \le  x_{ij} \le 1, && ~~\forall (ij)\in E. \label{cons:eb}
\end{alignat}
\endgroup

We refer to the above programs as \pra and \prb, respectively. Let  $\x_t^*$ and $\x^*_s$ be optimal solutions to \pra and \prb, respectively. 
 \begin{lemma}\label{lem:lp_a}
$\pi(\x_t^*)$ and $\pi(\x_s^*)$ are optimal policies  under \obja and \objb, respectively. 
\end{lemma}
 \begin{proof}
We focus on showcasing the case of $\obja$ and the program $\pra$. The proof for the other case is similar.  Note that the term shown in~\eqref{obj_a} captures the precise objective we aim to optimize. To prove our claim, we need to demonstrate that all constraints in~$\pra$ hold true for any viable policy of $\pi(\x)$. Constraint~\eqref{cons:j} is reasonable because every policy must assign each incoming task to a feasible worker without rejection, thereby ensuring that the total percentages assigned for each type sum up to one. Constraint~\eqref{cons:i} is valid as the workload of any worker (i.e., the percentage of busy time) should not exceed one. Constraint~\eqref{cons:e} holds true since $x_{ij}$ represents the percentage of tasks of type $j$ assigned to worker $i$.
   \end{proof}

Lemma~\ref{lem:lp_a} suggests that the optimal policies for $\obja$ and $\objb$ each can be obtained by solving minmax programs represented by $\pra$ and $\prb$ respectively. Note  that $\prb$ can be reformulated as a linear program (\LP) by introducing an auxiliary variable $\rho$ and modifying the objective as $\min \rho$, along with additional constraints $\rho \ge \rho_i$ for all $i \in I$. Consequently, we can efficiently solve $\prb$ and obtain an optimal policy for $\objb$. However, for 
program $\pra$, the objective is non-linear and can be neither convex nor concave even under very special settings, posing a technical challenge for direct optimization; 
\ifarxiv
see detailed discussions in the Appendix. 
\fi

Nevertheless, under certain conditions, $\pra$ can be effectively and accurately approximated by $\prb$, as proven in Theorem~\ref{thm:main}.



\begin{lemma}
\label{lem:lam}
The optimal values of $\pra$ and $\prb$ each remain invariant if we treat any task type $j\in J$ with an arrival rate of $\lam_j$ as $k$ different online types, each having the same set of neighbors as $j$, with an arrival rate of $\lam_j/k$ for any integer $k$.
\end{lemma}

The above lemma suggests that for fairness maximization among either workers under metric $\objb$ or tasks under metric $\obja$, we can assume without loss of generality that all tasks take a uniform arrival rate by creating an appropriate number of copies for each task type. In other words, the variation among tasks' arrival rates makes no difference to fairness promotion, compared with the difference among service times. \emph{\tbf{In the remaining  sections, we assume without loss of generality that $\lam_j=\lam$ for all $j \in J$.}}

\section{The Relation Between the Two Fairness Optimization Problems}\label{sec:cont}

Consider a general setting denoted by $\bmu:=(\mu_{ij})$, where $\mu_{ij}$ with $(ij) \in E$ represents the parameter for the exponential distribution of the service time taken by worker $i$ to serve task $j$. Let $\eta_t(\bmu, \x)$ denote the objective value of $\pra(\bmu)$ with respect to the input $\bmu$ and a feasible solution $\x=(x_{ij})$. Similarly, $\eta_s(\bmu,\x)$ denotes the objective value of $\prb(\bmu)$. When the context is clear, we may omit either the first or second argument for $\eta_t$ and $\eta_s$. For any given input $\bmu$, let $\eta^*_t(\bmu)$ and $\eta_s^*(\bmu)$ denote the optimal values of $\pra(\bmu)$ and $\prb(\bmu)$ respectively.

\begin{theorem}\label{thm:main}
Let $\x^*_s$ be an optimal solution to $\prb(\bmu)$. We have 
\begin{align}
\eta_t(\bmu, \x_s^*) \le \kap^3 \bP{1+\bp{1-\frac{1}{\kap}} \cdot \frac{\eta^*_s(\bmu)}{1-\eta^*_s(\bmu)} } \cdot \eta_t^*(\bmu), \label{ineq:thm-main}
\end{align}
where $\kap=\max_{i \in I} \sbp{\max_{j \sim i,j'\sim i}\mu_{ij}/\mu_{i,j'}} \ge 1$, which captures the maximum pairwise ratio among the expectations of all service time on each given worker. 
\end{theorem}


For a private case of Theorem~\ref{thm:main}- where $\kappa=1$ we prove that $\eta_t(\bmu, \x_s^*)=\eta_t^*(\bmu)$ (Theorem~\ref{thm:main-2}).

\ifarxiv
These results serve as the bedrock of the whole proof for Theorem~\ref{thm:main}, which is deferred to the Appendix for space reasons. 
\fi
Toward the proof of $\kappa=1$, we first define the following minimax programs and show  their equivalence to $\pra$ and $\prb$, respectively, for $\kappa=1$.

\begingroup
\allowdisplaybreaks
\begin{alignat}{2}
(\obar{\pra})~~\min &~~\max_{j \in J} \bp{\bw_j=\sum_{i \sim j} \frac{x_{ij}}{1-\rho_i}-1} ,&&  \label{obj_t2} \\
 & \sum_{i \sim j} x_{ij} =1,  &&~~\forall j \in J \label{t2cons:j}\\ 
 &  \rho_i=(\lam/\mu_i) \sum_{j \sim i} x_{ij} \le 1,&&~~ \forall i \in I \label{rho1}\\ 
  &0 \le  x_{ij} \le 1, && ~~(ij)\in E. \label{t2cons:e} 
\end{alignat}
\endgroup
\begin{alignat}{2}
(\obar{\prb})~~\min &~~\max_{i \in I} \rho_i ,&&  \label{obj_s2} \\
 & \sum_{i \sim j} x_{ij} =1,  &&~~\forall j \in J  \\ 
 &  \rho_i= (\lam/\mu_i)\sum_{j \sim i} x_{ij} \le 1,&&~~ \forall i \in I  \label{rho2} \\ 
  &0 \le  x_{ij} \le 1, && ~~(ij)\in E.  
  \end{alignat}

\begin{lemma}
For $\kappa=1$, the programs $\pra$ and $\obar{\pra}$ are equivalent and the programs $\prb$ and $\obar{\prb}$ are also equivalent.
\end{lemma}  

\begin{proof}
Note that $\kap=1$ suggests that $\mu_{ij}$ takes some uniform value of $\mu_{ij}=\mu_i$ for every $j \sim i$. Recall that $\lam_j=\lam$ for all $j \in J$ due to Lemma~\ref{lem:lam}. Under these assumptions, we see that the expressions of $\rho_i$ and $\bw_j$ in~\eqref{eqn:rho_i} and~\eqref{eqn:bw_j} can be simplified as
\begin{align*}
&\rho_i=\sum_{\ell \sim i} x_{i \ell} \lam_\ell/\mu_{i \ell}=(\lam/\mu_i) \cdot \sum_{\ell \sim i} x_{i \ell},\\
&\bw_j:=\E[\bW_j]=\sum_{i \sim j} x_{ij} \cdot w_i/(1/\mu_{ij}) \\
&=\sum_{i \sim j} \frac{x_{ij} \cdot \mu_{i} \cdot \sum_{\ell \sim i} x_{i \ell} \lam_\ell/\mu_{i}^2}{1-\sum_{\ell \sim i} x_{i \ell} \lam_\ell/\mu_{i}}=\sum_{i \sim j}   \frac{x_{ij} \cdot\sum_{\ell \sim i} x_{i \ell} \lam_\ell/\mu_{i}}{1-\sum_{\ell \sim i} x_{i \ell} \lam_\ell/\mu_{i}}\\
&=\sum_{i \sim j} x_{ij}\bp{-1+\frac{1}{1-\rho_i}}=-1+\sum_{i \sim j} \frac{x_{ij}}{1-\rho_i},
\end{align*}
where the equality on the last line is due to $\sum_{i \sim j} x_{ij}=1$ for every $j \in J$ (no rejection allowed). 
Substituting lines \ref{rho1} and \ref{rho2} with the value of $\rho_i$ and line \ref{obj_t2} with the value of $\bw_j$ implies that the programs \pra and \prb are equivalent to \obar{\pra} and \obar{\prb}.
\end{proof}

Consider a given setting with $\bmu=(\mu_{ij})$ satisfying $\mu_{ij}=\mu_i$ for all $j \sim i$ and $\lam_j=\lam$ for all $j \in J$. For ease of notation, we use  
$\obar{\eta^*_t}$ and $\obar{\eta^*_s}$ to denote optimal values of \obar{\pra} 
and \obar{\prb}, respectively, with respect to the given setting. By default, we assume both have feasible solutions.\footnote{Infeasibility to either Program $\obar{\pra}$ or $\obar{\prb}$ suggests that no policy can lead to meaningful fairness among tasks (finite max expected waiting time) or among workers (a non-zero ratio of being free).} 
We denote by $\obar{\eta_t(\x_s^*)}$ the value of $\emph{\obar{\pra}}$ on $\x_s^*$ in the given setting.

It is tempting to prove that $\eta_t(\x_s^*)=\eta_t^*$ for $\kappa=1$ by showing that  \obar{\pra} and  \obar{\prb} each possess an optimal solution such that $\{\rho_i\}$ all take a uniform value, say $\rho$. Following this ``claim'',  \obar{\pra} is then reduced to $\min 1/(1-\rho)-1$ with $\rho=\rho_i$ for all $i \in I$, while \obar{\prb} is reduced to $\min  \rho$ with $\rho=\rho_i$ for all $i \in I$. This establishes Theorem~\ref{thm:main-2} since $\min 1/(1-\rho)-1$ is equivalent to $\min \rho$. 
The example below disproves this idea, unfortunately.

\begin{example}\label{exam}[\tbf{\emph{\obar{\pra}} and \emph{\obar{\prb}} each
possess a unique optimal solution with non-uniform values of $\{\rho_i\}$ and $\{\bw_j\}$}.]
Consider a graph $G=(I,J,E)$ such that $|I|=m=2$ and $|J|=n \gg 1$ 
\ifarxiv
(See Figure~\ref{fig:exam-1}).
\fi
The input setting is as follows. $\mu_{ij}=\mu$ for all $(ij) \in E$ and $\lam_j=\lam$ for all $j \in J$. Let $\phi=\lam/\mu$ with $n \cdot \phi<1$. $i=2$ is connected to all $j \in J$, while $i=1$ is connected only to $j=1$. We can verify that (1) \emph{\obar{\pra}} and \emph{\obar{\prb}} each have a  \emph{unique} optimal solution and the two are the same, which is $\x^*=(x_{ij})$ with $x_{11}=1$, $x_{21}=0$, and $x_{2j}=1$ for all $1<j \le n$; (2) for \emph{\obar{\prb}}: $\rho_1(\x^*)=\phi$, and $\rho_2(\x^*)=(n-1)  \phi<1$; for \emph{\obar{\pra}}: $\bw_1(\x^*)=1/(1-\rho_1(\x^*))-1=1/(1-\phi)-1$ and $\bw_j =1/(1-\rho_2(\x^*))-1=1/(1-(n-1)\phi)-1$ for $1<j \le n$.
\end{example}

We will now present two lemmas that establish together the correctness of Theorem~\ref{thm:main-2}.
\begin{lemma}\label{lem:<=}

$\obar{\eta_t(\x_s^*)} \le  1/(1-\obar{\eta_s^*})-1$ 
\end{lemma}

\begin{proof}
Since $\x^*=(x_{ij})$ is an optimal solution to \obar{\prb}, $\obar{\eta_s^*}=\max_{i \in I} \rho_i(\x^*):=\rho^*$. Observe that for each $j \in J$,
\[
\bw_j (\x^*)=\sum_{i \sim j} \frac{x_{ij}}{1-\rho_i(\x^*)} \le \sum_{i \sim j} \frac{x_{ij}}{1-\rho^*}=\frac{1}{1-\rho^*}-1,
\]
which suggests that  $\obar{\eta_t(\x^*)}=\max_j \bw_j(\x^*) \le 1/(1-\rho^*)-1=1/(1-\obar{\eta^*_s})-1$. 
\end{proof}
\ifarxiv
\begin{figure}[th!]
\begin{minipage}{.5\linewidth}
 \begin{tikzpicture}
 \draw (0,0) node[minimum size=0.2mm,draw,circle] {$i_1$};
  \draw (0,-1) node[minimum size=0.2mm,draw,circle] {$i_2$};
  
  \draw (3,0) node[minimum size=1mm,draw,circle] {$j_1$};
    \draw (3,-1) node[minimum size=1mm,draw,circle] {$j_2$};

            \draw (3,-3) node[minimum size=1mm,draw,circle] {$j_n$};
\draw[ultra thick, -] (0.4,0)--(2.6,0);
\draw[ultra thick, -] (0.4,-1)--(2.6,-1);

\draw[ultra thick, -] (0.4,-1)--(2.6,0);
\draw[ultra thick, -] (0.4,-1)--(2.6,-1);
\draw[ultra thick, -] (0.4,-1)--(2.6,-3);
\draw[ultra thick, -] (0.4,-1)--(2.6,-2);
\draw [dotted, ultra thick](2.7,-2) -- (3.3,-2);
\end{tikzpicture}
\end{minipage}%
\hfill
\begin{minipage}{.5\linewidth}
\begin{align*}
& \rho_1(\x^*) =\phi,\\
&\rho_2(\x^*)=(n-1)  \phi,\\
& \bw_1(\x^*)=\frac{1}{1-\phi}-1,\\
& \bw_{j}(\x^*)=\frac{1}{1-(n-1)\phi}-1.
\end{align*}
\end{minipage}%
\caption{An example where  {\obar{\pra}} and {\obar{\prb}} each have a unique optimal solution with non-uniform values of $\{\rho_i\}$ and $\{\bw_j\}$, though $\eta_t(\bmu,\x_s^*)=\eta^*_t$ since the programs share the same unique optimal solution.} \label{fig:exam-1}
\end{figure}
\fi
\begin{lemma}\label{lem:>=}
$1/(1-\obar{\eta_s^*})-1 \le \obar{\eta_t^*}$.
\end{lemma}
\ifarxiv
The lemma's proof is in the Appendix.
\fi 

We're now set to present results for $\kap=1$.

\begin{theorem}
\label{thm:main-2}
Consider an input $\bmu=(\mu_{ij})$ with $\kap=1$. Let  $\x_s^*$ be an optimal solution to $\emph{\obar{\prb}}$.  We have that the value of $\emph{\obar{\pra}}$ on the solution of $\x_s^*$ is equal to its optimal value, \ie $\eta_t(\x_s^*)=\eta_t^*$.
\end{theorem}

\begin{proof}
    The above two lemmas together imply that $\obar{\eta_t(\x_s^*)} \le \obar{\eta_t^*}$. $\x_s^*$ is feasible to $\obar{\pra}$ since $\obar{\pra}$ and $\obar{\prb}$ share the same set of constraints, and thus, $\obar{\eta_t(\x_s^*)} \ge \obar{\eta_t^*}$, which establishes Theorem~\ref{thm:main-2}. 
\end{proof}
\section{Experiments}

\subsection{Algorithms and Heuristics}

This section presents an algorithm derived from solutions to one of the minimax problems. 
We also describe a heuristic based on this algorithm, which gives preference to assigning tasks to available workers, thereby enhancing allocation through the effective workload of free workers.  In addition, this section introduces two real-time greedy heuristics, which function as baseline methodologies.

\paragraph{Minimax problems based algorithm:}
We first describe Algorithm ~\ref{alg:pb}. This algorithm has offline and online phases. 
In the offline phase (line 2), a solution to one of the minimax problems is computed. In the online phase (lines 4-7), when a task arrives, the task is assigned to the queue of a worker according to the probabilities computed by the program in the offline phase. This algorithm has two variants: One solves \pra in the offline phase while the other solves \prb.

\paragraph{Minimax problems based heuristic:}
A notable issue with Algorithm 1 is that tasks can wait for a busy worker despite other available workers. This leads to suboptimal performance. To address this, we create a heuristic based on Algorithm~\ref{alg:pb}. Like Algorithm~\ref{alg:pb}, in Algorithm~\ref{alg:ha}, task assignment probabilities are computed offline to mitigate this problem. In the online phase, incoming tasks are assigned to free workers. If multiple workers are free, their precomputed probabilities (from the offline phase) are normalized to sum to 1. A worker is subsequently chosen randomly, guided by these normalized probabilities.
If there are no free workers, the tasks are assigned according to their probabilities as in Algorithm~\ref{alg:pb}. Algorithm \ref{alg:ha} describes this heuristic. As in Algorithm \ref{alg:pb}, there are two variants of Algorithm \ref{alg:ha}: One solves ~$\pra$ in the offline phase, while the other solves $\prb$.

This method targets reduced waiting times, especially during low-load periods. However, this change might decrease worker workload or waiting times for other tasks, as it deviates from calculated optimal probabilities. In practice, we find that the trade-off for worker and task fairness is reasonable, given the substantial benefits for all tasks' fairness. 

\paragraph{Computational complexity of Algorithms~\ref{alg:pb} and ~\ref{alg:ha}}
Both algorithms~\ref{alg:pb} and ~\ref{alg:ha} have \textbf{offline} and \textbf{online} phases.
The offline phase is identical for both algorithms and requires the solution of \pra or \prb.
Following ~\cite{cohen2021solving},
the runtime for solving the linear program- \prb can be as low as $O^{*}(N^{2+1/6}\log(N/\delta))$,
where $\delta$ is the relative accuracy and $N = |E|$ is the number of edges in the graph $G$. We leave the complexity of solving $\pra$ to future work.
In any case, the complexity of the offline phase dominates the complexity of the online phase.

\begin{algorithm}[th!]
\DontPrintSemicolon
\textbf{Offline Phase}: \;
Solve \pra (\prb) and let $\{x_{ij}\}$ be an optimal solution.\;
\textbf{Online Phase}:\;
\For{each task of type $j$ that arrive on time $t$}
{
Let $\cQ_i$ be a queue of worker $i$\;
Choose randomly a worker $i$ following the probabilities in $\{x_{ij}\}$\;

Update $\cQ_i=\cQ_i \cup \{(j,t)\}$. \;}

\caption{An LP-based algorithm for \obja and \objb.}
\label{alg:pb}
\end{algorithm}

\paragraph{The two greedy heuristics}
Similar to~\cite{viden}, we devised two greedy heuristics as baselines for comparison. The first minimizes maximum task waiting times, and the second minimizes maximum worker workload.
In the first, incoming tasks are assigned to workers with the shortest estimated waiting time, calculated by summing average expected task durations for tasks in the queue. The elapsed time for ongoing tasks is subtracted from their average duration to update estimates.
For the second, tasks are assigned to less utilized workers based on current workload upon task arrival. This approach considers executed tasks, using actual durations rather than expected durations. 
The first heuristic is denoted as GTW (Greedy Task Waiting time) and the second as GWU (Greedy Worker Workload).

\subsection{Experimental Settings}
We ran experiments on the teleoperation domain. As already mentioned in Section~\ref{sec:intro},
the teleoperation of AVs involves intervention tasks that are assigned to the teleoperators who perform them.
We adapted the dataset of \cite{viden} for our two-sided fairness study.
\ifarxiv
More details about the experimental settings and additional experimental results have been moved to the Appendix. 
\fi
Source code and data for running the experiments are available at \cite{two_sided_fairness}.

\begin{figure*}
    \centering
    \includegraphics[width=0.6\columnwidth]{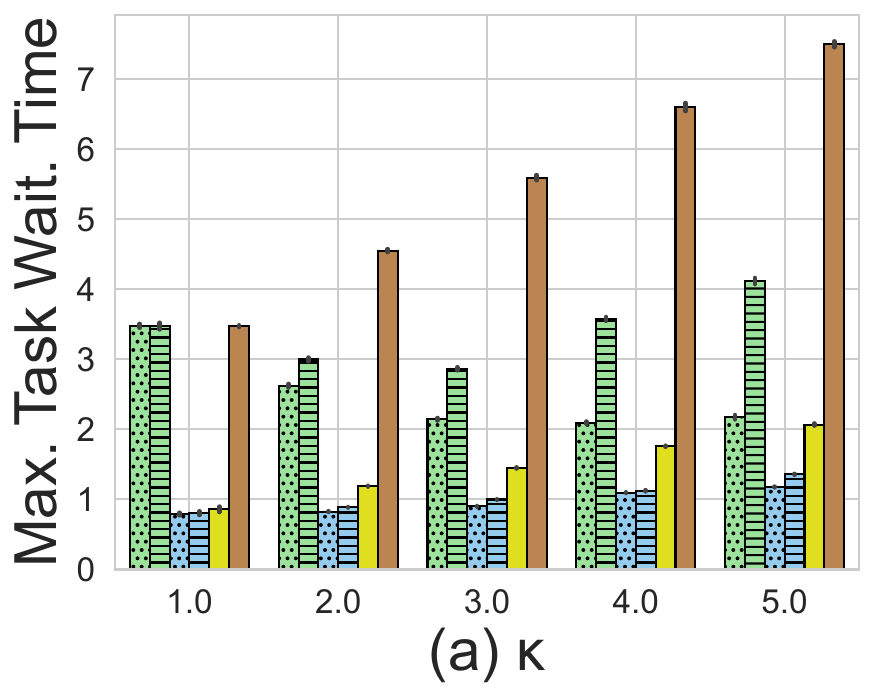}
    \includegraphics[width=0.6\columnwidth]{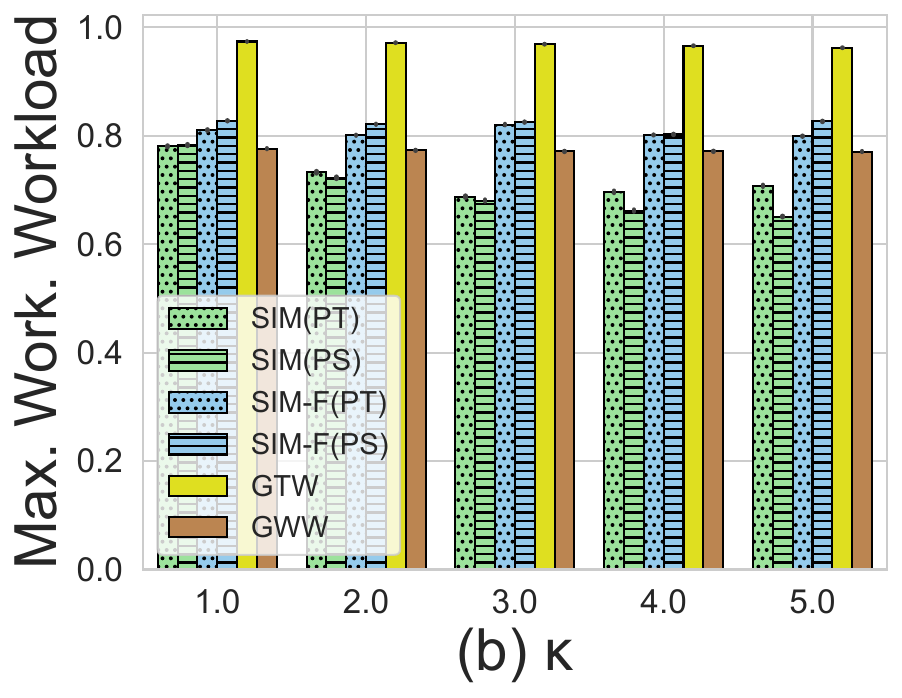}
    \includegraphics[width=0.6\columnwidth]{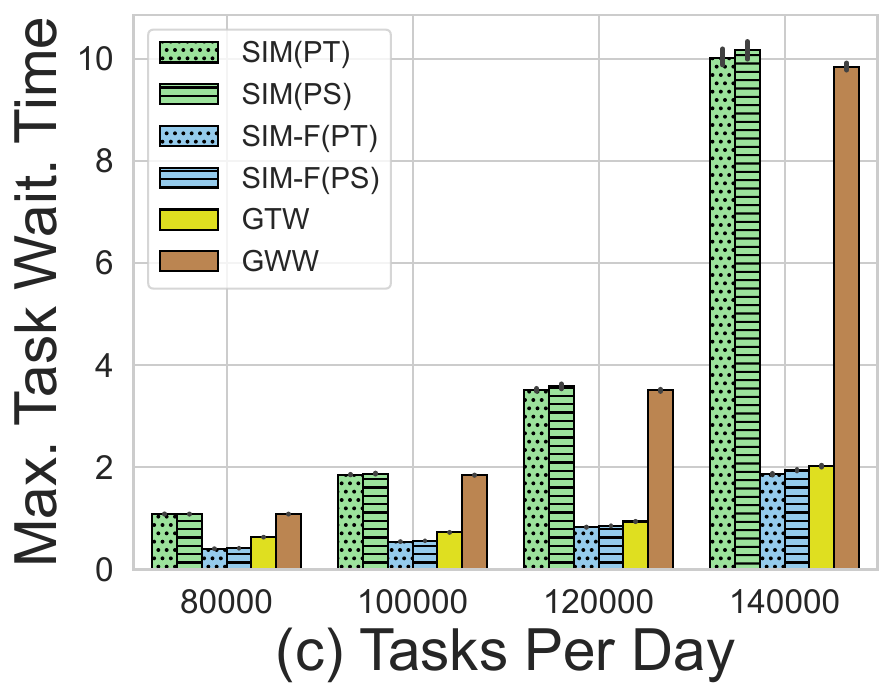}
    \includegraphics[width=0.6\columnwidth]{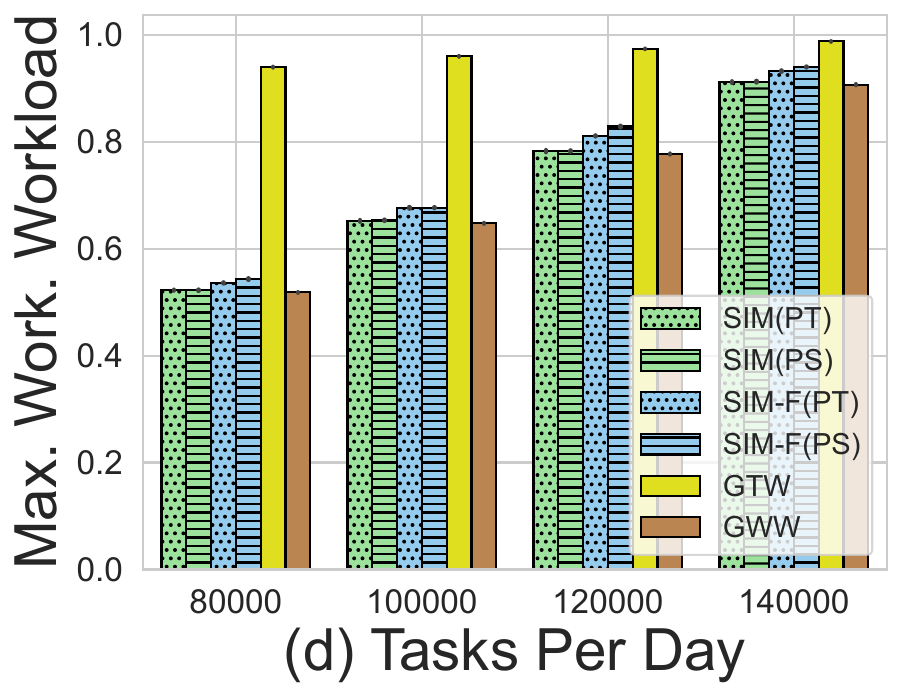}
    \includegraphics[width=0.6\columnwidth]{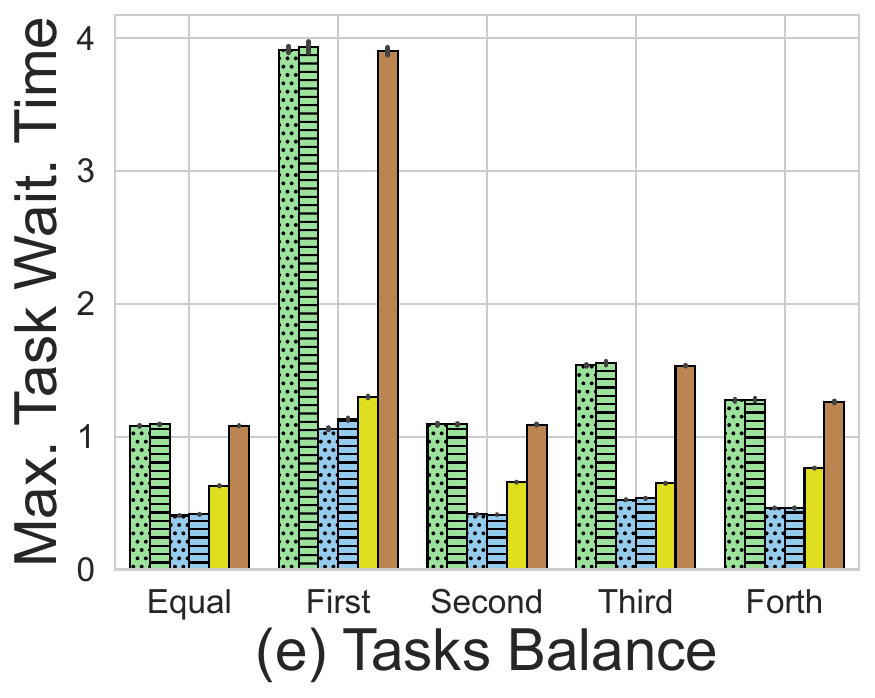}
    \includegraphics[width=0.6\columnwidth]{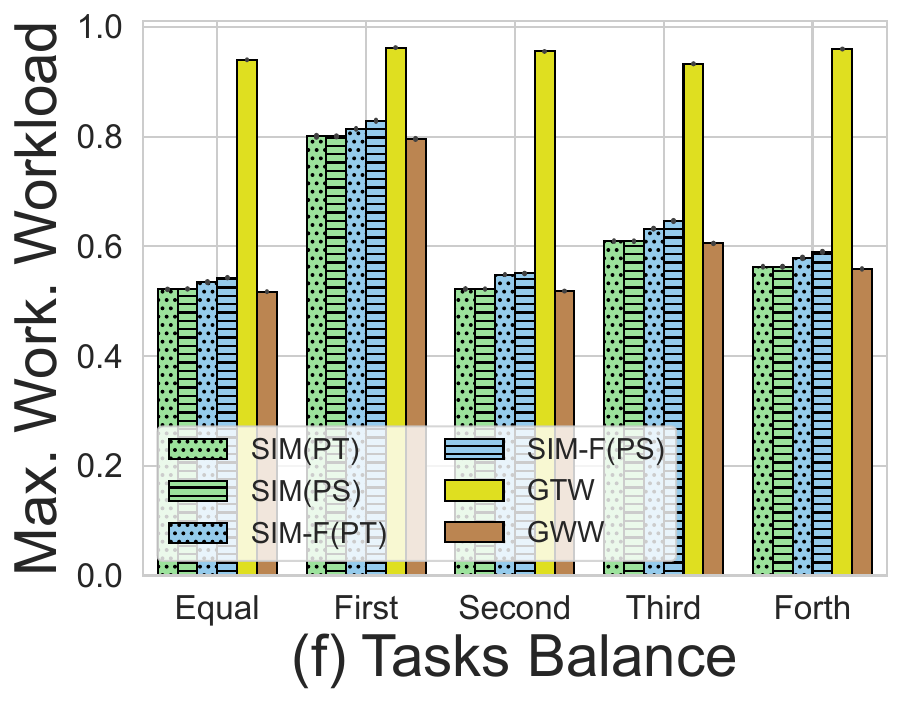}
    \caption{ Y axis is the maximum waiting time(in seconds) for a task(a,c and e) and the max. worker workload(b,d and f). The X axis in (a,b) is the value of $\kappa$(x axis) and the task load is of 120000 tasks per day. In (c,d), the X axis is the task load and $\kappa$ is set to 1. In (e,f) the X axis is for different balances of arrival distribution: first bar is for equal distribution for each task type. In the second bar, the first type has probability of 70\% and the others have 10\%. the other bars are defined similarly for the second, third and forth task types (task load was 80000 tasks per day and $\kappa$ is set to 1).
    In the legend: SIM(PT) and SIM(PS) denote Algorithm~\ref{alg:pb}'s results for \pra and \prb in simulation. SIM-F(PT) and SIM-F(PS) are Algorithm~\ref{alg:ha}'s results (in which we assign to a \textbf{free} worker first) for $\pra$ and $\prb$ in simulation. $GTW$ and $GWU$ are the results of the greedy heuristics targeting task waiting time and worker workload in simulation.
    Error bars represent a confidence interval of 0.95.
    }
    \label{fig:window}
\end{figure*}

\paragraph{The tasks, their durations and their arrival rates}

 Our study built upon the four task types defined by Viden et al.~(\citeyear{viden}). Their dataset provided valuable insights into the average duration times for each teleoperator (worker) and task type in a simulation.
We explored three distinct approaches to define task duration in our experiments. All approaches involved sampling durations from exponential distributions, but the difference lay in the means of these distributions.

\paragraph{The teleoperators and the tasks they can perform}
Using the dataset of Viden et al.~(\citeyear{viden}), we initially had 10 teleoperators (workers) and 4 task types. 
We form a bipartite graph with 10 teleoperators on one side and 4 task types on the other. The dataset provides average task completion times for each teleoperator-task pair. An edge is established between a teleoperator and a task type if their average time matches or exceeds the task type's median value. Following this process, a teleoperator who consistently performed tasks slower than the median was identified and subsequently excluded from the graph. 
\ifarxiv
More experiments on a synthetic dataset in which the numbers of teleoperators and task types are varied can be found in the Appendix. 
\fi

\begin{algorithm}[th!]
\DontPrintSemicolon
\textbf{Offline Phase}: \;
Solve \pra (\prb) and let $\{x_{ij}\}$ be an optimal solution.\;
\textbf{Online Phase}:\;
\For{each task of type $j$ that arrive on time $t$}
{
Let $\cQ_i$ be a queue of worker $i$ and let $F$ be the subset of free workers on time $t$\;
If $F\neq \emptyset$, randomly choose a free worker $i$ with probability $\{x_{ij}/\sum_{i' \in F}{x_{i'j}}\}$\;
Otherwise choose randomly a worker $i$ following the probabilities in $\{x_{ij}\}$\;
Update $\cQ_i=\cQ_i \cup \{(j,t)\}$. \;}
\caption{A heuristic for \obja(\objb).}
\label{alg:ha}
\end{algorithm}

\paragraph{Experimental environment and more settings}
Each experiment spanned a virtual 4-week period. Due to algorithmic stochasticity, each experiment was repeated 10 times.

\subsection{Results and Discussion}

\paragraph{Effect of changing $\kappa$}
In Figures~\ref{fig:window}(a,b), we illustrate the performance of various methods across diverse $\kappa$ values. 
In Figure ~\ref{fig:window}(a), we measure the maximum task waiting time. We see that the gap between SIM(\prb) and  SIM(\pra), as well as 
the gap between SIM-F(\prb) and SIM-F(\pra), increase with $\kappa$.
This aligns with the fact that with higher values of $\kappa$, the approximation ratio of \prb's solution relative to \pra's objective is greater. However, the ratio between the different methods measured in practice is lower than the worst-case theoretical ratio given by Theorem~\ref{thm:main} (which is greater than $\kappa^3$). 

In Figure~\ref{fig:window}(b), we measure the maximum worker workload. The differences between SIM(\pra) vs SIM(\prb) are very small for $\kappa\leq 3$, but they become more significant for $\kappa\in\{4,5\}$.
Surprisingly, there is a different effect with SIM-F(\pra) and SIM-F(\prb). SIM-F(\pra) performs slightly better than SIM-F(\prb). We conjecture that the initial selection of free workers has a more detrimental effect in SIM-F(\prb), which integrates two distinctly different methods, in contrast to the relatively similar approaches in 
SIM-F(\pra). We also see that for larger values of $\kappa$, both SIM(\pra) and SIM(\prb) perform worse than for lower values.


\paragraph{Effect of changing the task load}
Figures~\ref{fig:window}(c,d) might help the teleoperation center's owner decide whether the current number of workers is sufficient. It is noticeable that in Figure~\ref{fig:window}(c) there is a significant jump from 120000 to 140000 tasks per day. This means that perhaps the owner should employ more workers in this case. 
Referring to Figure 2(d) may lead us to similar conclusions. Employing more workers is advisable if individual worker workload is excessively high.

\paragraph{Effect of changing the task balance}
Figures~\ref{fig:window}(e,f) represent the performance of the different algorithms when changing the task balance.
The left bar represents an even distribution for each task type (0.25). The second bar represents a higher probability for the first type (0.7) and a lower probability for the other types (0.1). The other bars are similarly defined for the other task types.

In Figures~\ref{fig:window}(e,f), higher arrival distribution of the first task type leads to elevated waiting times and worker workload. Consequently, the teleoperation center's owner could enhance fairness by upskilling operators who are not qualified for the task or hiring new ones proficient in it. Alternatively, training could be provided to expedite task completion.
The negligible error bars in all figures show that the error approaches 0 if the experiments are carried out over a sufficiently long period of time, as we have done.

\paragraph{Computed optimal values vs simulation values}
In all experiments that we ran, the computed expected maximum waiting time (OPT(\pra)) and the computed expected maximum worker workload (OPT(\prb)) closely align with simulation-derived values (SIM(\pra) and SIM(\prb) respectively). Additionally, the alignment of OPT(\pra) and OPT(\prb) at $\kappa=1$ is consistent with Theorem~\ref{thm:main}. 

\paragraph{Choosing the best algorithm}

The heuristic GTW, which minimizes the maximum task waiting time, performs well at maximum task waiting time and performs poorly at maximum worker workload. Conversely, the greedy heuristic that minimizes the maximum worker workload, GWU, performs well at the maximum worker workload and performs poorly at maximum task waiting time. 
The methods that offer the best tradeoff between two dimensions of fairness are SIM-F(\pra) and SIM-F(\prb). However, since \pra is nonlinear, there is no tool that guarantees to find an optimal solution for $\pra$, and therefore $\pra$-dependent approaches such as SIM-F(\pra) might be unsolvable.

Therefore, if $\kappa=1$ or at least a small number close to $1$, we might want to use SIM-F(\prb).
However, SIM(\prb) might be slightly better if worker workload is more important than task waiting times (but still important).
If $\kappa$ is large, it is advisable to consider using a tool that approximates a solution for \pra with SIM-F(\pra). The figures show that the available tools work adequately in such cases, despite the lack of theoretical guarantees (at least for small problems). Another option is to try both SIM-F(\pra) and SIM-F(\prb) and pick the one that gives the best results. 

\section{Conclusion}

This paper addresses two-sided fairness problems represented as online bipartite matching with accommodated delays. We introduce two minimax problems: \pra to minimize the maximum workload of workers and \prb to minimize the maximum waiting time of tasks. We show that the second problem can be formulated as a linear program and thus solved efficiently. Moreover, we showed that the policy using a solution for \prb approximates the solution for \pra, and we then presented an upper bound on the approximation ratio. Finally, we compared the performance of different approaches (most of them used the solutions to the problems) and empirically evaluated their performance.

Future research may explore different definitions of fairness. In addition, it is promising to extend our approach to scenarios where workers are also arriving dynamically. To demonstrate the need in such scenarios one might consider the teleoperation application where teleoperators (workers) can join or leave the crew. 
Considering different distributions for both task arrivals and task durations can provide more depth and insights into the study.
Finally, it might be beneficial to consider some robust version, say, minimization of the maximum possible absolute waiting time among users, which is equivalent to the minimization of the maximum absolute waiting time among all workers.

\section*{Acknowledgements}
This research has been partially supported by the Israel Science
Foundation under grant 1958/20 and  the EU Project TAILOR
under grant 952215.
Work of Pan Xu was partially supported by
NSF CRII Award IIS-1948157.

\bibliographystyle{named}
\bibliography{refs}

\ifarxiv
\clearpage
 \onecolumn
 \appendix
\begin{center}
\fbox{{\Large\textbf{Technical Appendix}}}
\end{center}

\bigskip

\noindent
\textbf{Paper title:}~Design a Win-Win Strategy That Is Fair to Both Service Providers and Tasks When Rejection is Not an Option

\smallskip
\noindent
\textbf{Paper id:}~  185
\medskip

\section{Objective Function in Program \pra Can Be Neither Convex nor Concave (Even When $\kap=1$)}
We show that the optimization program \pra can be minimization of a non-convex function.  Consider the example shown in Figure~\ref{fig:exam-2}, where $\lam_j=\lam$ for all $j \in J$, and $\mu_{ij}=\mu_i, \forall j \sim i, \forall i \in I$. Set  $\phi_i=\lam/\mu_i$, for each $i \in I$. Let $x$ be the value on edge $(i=1,j=1)$, and thus, $1-x$ be that on edge $(i=2, j=1)$. Similarly, let $y$ and $1-y$ be the values on edges $(i=3,j=2)$ and $(i=2,j=2)$. We can verify that
\begin{align*}
\rho_1&=\phi_1 \cdot x, \rho_2=\phi_2 (1-x+1-y), \rho_3=\phi_3 \cdot y;\\
\bw_1 &=\frac{x}{1-\phi_1 \cdot x}+\frac{1-x}{1-\phi_2 (1-x+1-y)}-1\\
\bw_2&=\frac{y}{1-\phi_3 \cdot y}+\frac{1-y}{1-\phi_2 (1-x+1-y)}-1.
\end{align*}
Let $\bph=(\phi_1,\phi_2,\phi_3)$, and assume $\phi_1 \le 1, \phi_3 \le 1, \phi_2 \le 1/2$. Under these assumptions, we see the feasible region of \pra can be reduced to $\Omega=\{(x,y): 0 \le x, y\le 1\}$.
\begin{align*}
f_{\bph}(x,y)&=\max(\bw_1, \bw_2)=\max\bP{\frac{x}{1-\phi_1 \cdot x}+\frac{1-x}{1-\phi_2 (1-x+1-y)}-1,\frac{y}{1-\phi_3 \cdot y}+\frac{1-y}{1-\phi_2 (1-x+1-y)}-1}\\
\Omega&=\{(x,y): 0 \le x, y\le 1\}.
\end{align*}

Furthermore, set $\phi_1=\phi_3=p \in [0,1]$ and $\phi_2=q \in [0,1/2]$. Define 
\begin{align*}
g_{p,q}(x,y)=\frac{x}{1-p \cdot x}+\frac{1-x}{1-q (2-x-y)}-1, ~~h_{p,q}(x,y)=\frac{y}{1-p \cdot y}+\frac{1-y}{1-q (2-x-y)}-1.
\end{align*}
Thus, the original program \pra can be simplified as
\begin{align}
 \min f_{p,q}(x,y):=\max\bp{g_{p,q}(x,y), h_{p,q}(x,y)}:~~0 \le x, y\le 1.
\end{align}
By symmetry, we can assume WLOG that $x \ge y$. Then we see
\begin{align*}
f_{p,q}(x,y)&=h_{p,q}(x,y)=\frac{y}{1-p \cdot y}+\frac{1-y}{1-q (2-x-y)}-1, && \mbox{~if~} 2q+p^2 xy \ge (p+q) \cdot (x+y), ~0 \le y \le x \le 1;\\
&=g_{p,q}(x,y)=\frac{x}{1-p \cdot x}+\frac{1-x}{1-q (2-x-y)}-1, &&\mbox{~if~} 2q+p^2 xy \le (p+q) \cdot (x+y), ~0 \le y \le x \le 1.
\end{align*}
We can verify that the region $\Omega_1:=\{(x,y): 2q+p^2 xy \ge (p+q) \cdot (x+y), ~0 \le y \le x \le 1\}$ is convex, while the function $f_{p,q}(x,y)=h_{p,q}(x,y)$ is \emph{neither convex nor concave} over $\Omega_1$ under many settings of $(p,q)$, say, \eg $p=0.7$ and $q=0.4$. This establishes $f_{p,q}(x,y)$ is neither convex or concave over the original region $\Omega=\{(x,y): 0 \le x, y\le 1\}$.

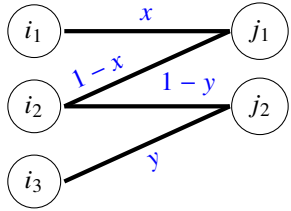
\begin{figure}[th!]
\center
 \begin{tikzpicture}
  \draw (0,0) node[minimum size=0.2mm,draw,circle] {$i_1$};
  \draw (0,-1) node[minimum size=0.2mm,draw,circle] {$i_2$};
    \draw (0,-2) node[minimum size=0.2mm,draw,circle] {$i_3$};
  \draw (3,0) node[minimum size=1mm,draw,circle] {$j_1$};

    \draw (3,-1) node[minimum size=1mm,draw,circle] {$j_2$};

\draw[ultra thick, -] (0.4,0)--(2.6,0) node [blue, above, midway, sloped] {$x$};
\draw[ultra thick, -] (0.4,-1)--(2.6,-1) node [blue, above, near end, sloped] {$1-y$};

\draw[ultra thick, -] (0.4,-1)--(2.6,0) node [blue, above, near start, sloped] {$1-x$};

\draw[ultra thick, -] (0.4,-2)--(2.6,-1) node [blue, below, midway, sloped] {$y$};

\end{tikzpicture}
\caption{A toy example showing the objective function in Program \pra can be neither convex nor concave even when $\kap=1$.} \label{fig:exam-2}
\end{figure}

\section{Proof of Lemma~\ref{lem:lam}}
Consider a given input setting characterized by $\bmu:=(\mu_{ij})$ and $\blam:=(\lam_{j})$. Recall that $\eta_t^*(\blam,\bmu)$ and $\eta_s^*(\blam,\bmu)$ denote the optimal values of $\pra$ and $\prb$ with respect to $\blam$ and $\bmu$. We focus on the case when $k=2$, and all the analysis can be straightforwardly generalized to any generic integer $k$. Consider a modified version of $\blam$ in which the first request type $j=1$ is split into two copies, $j=0$ and $j=1$, each having an arrival rate of $\lam_1/2$. Let $\bola=(\balam_j)$ be this modified version with $\balam_0=\balam_1=\lam_1/2$ and $\balam_j=\lam_j$ for all $j \ge 2$. Let $\bar{J}=J \cup \{0\}$ be the resulting set of online types.

\begin{lemma}\label{lem:lama}
$\eta_t^*(\blam,\bmu)=\eta_t^*(\bola,\bmu)$, and $\eta_s^*(\blam,\bmu)=\eta_s^*(\bola,\bmu)$, 
\end{lemma}
\begin{proof}
 Let us focus on showing the first equality from Lemma~\ref{lem:lama}. Suppose $\x=(x_{ij})$ is an optimal solution to $\pra(\blam,\bmu)$. Consider a modified solution $\y=(y_{ij})$ that satisfies: (1) $y_{ij}=x_{ij}$ for all $j \ge 2$ and $i \sim j$, and (2) $y_{i,0}=y_{i,1}=x_{i,1}$. Let $\rho_i(\x)$ and $\rho_i(\y)$ be the values of $\rho_i$ with respect to $\x$ under the setting of $(\blam,\bmu)$ and $\y$ under the setting of $(\bola,\bmu)$, respectively. Similarly, let $\bw_j(\x)$ and $\bw_j(\y)$ represent the values of $\bw_j$ with respect to $\x$ and $\y$, respectively. We can verify that $\y$ is feasible for $\pra(\bola,\bmu)$. Furthermore, $\rho_i(\x)=\rho_i(\y)$ for all $i \in I$, and $\bw_j(\x)=\bw_j(\y)$ for all $j\ge 2$, and $\bw_0(\y)=\bw_1(\y)=\bw_{1}(\x)$. Thus, we claim that:
\[
\eta_t^*(\bola,\bmu) \le \max_{j \in \bar{J} } \bw_j(\y)=\max_{j  \in J } \bw_j(\x)=\eta_t^*(\blam,\bmu).
\]

Now we prove the other direction. Let $\y=(y_{ij})$ be an optimal solution to $\pra(\bola,\bmu)$. Consider a modified solution $\x=(x_{ij})$ that satisfies: (1) $x_{ij}=y_{ij}$ for all $j \ge 2$ and $i \sim j$, and (2) $x_{i,1}=(y_{i,0}+y_{i,1})/2$ for all $i \sim j$. We can verify that $\x$ is feasible for $\pra(\blam,\bmu)$. Furthermore, $\rho_i(\x)=\rho_i(\y)$ for all $i \in I$, and $\bw_j(\x)=\bw_j(\y)$ for all $j\ge 2$, and $\bw_1(\x)=(\bw_{0}(\y)+\bw_1(\y))/2$. Thus, we claim that:
\[
\eta_t^*(\blam,\bmu) \le \max_{j \in J } \bw_j(\x) \le \max_{j  \in \bar{J} } \bw_j(\y)=\eta_t^*(\bola,\bmu).
\]
This shows that the first equality holds, and the second equality follows from the same argument.
\end{proof}

\section{Proof of Theorem~\ref{thm:main}}

Now, we consider a general setting $\bmu=(\mu_{ij})$ with $\kap=\max_{i \in I} \max_{j \sim i, j' \sim i} \mu_{ij}/\mu_{i,j'} \ge 1$. Consider such a virtual instance that for each $i \in I$, all of $\mu_{ij}$ with $j \sim i$ are replaced with ${\mu}_i:=\max_{j \sim i}\mu_{ij}$. Let $\bomu=(\bamu_{ij})$ be the modified version of $\bmu$ such that $\bamu_{ij}=\mu_i$ for every $j \sim i$ and $i \in I$. By Claims~\ref{claim:main-3} and~\ref{claim:>=}, we see that $\eta_t^*(\bomu)=-1+1/(1-\eta_s^*(\bomu))$, where $\eta_t^*(\bomu)$ and $\eta_s^*(\bomu)$ denote the optimal values of \pra and \prb under $\bomu$, respectively. Similarly,  $\eta_t^*(\bmu)$ and $\eta_s^*(\bmu)$ denote the optimal values of \pra and \prb under $\bmu$, respectively.

\begin{lemma}\label{lem:19-a}
(1) $\eta_t^*(\bmu) \ge \eta_t^*(\bomu)/\kap$; (2) $\eta_s^*(\bomu) \ge \eta_s^*(\bmu)/\kap$.
\end{lemma}
\begin{proof}
We prove the first inequality as follows. 
Consider any optimal solution $\x=(x_{ij})$ to $\pra(\bmu)$. Observe that for any $j \in J$, 
\begin{align*}
\bw_j (\bmu,\x) &=\sum_{i \sim j} x_{ij} \cdot \mu_{ij} \cdot \frac{\sum_{\ell \sim i} x_{i \ell} \lam_\ell/\mu_{i \ell}^2}{1-\sum_{\ell \sim i} x_{i \ell} \lam_\ell/\mu_{i \ell}}\\
&\ge (1/\kap) \cdot \sum_{i \sim j} x_{ij} \cdot \frac{\sum_{\ell \sim i} x_{i \ell} \lam_\ell/\mu_{i \ell}}{1-\sum_{\ell \sim i} x_{i \ell} \lam_\ell/\mu_{i \ell}}=(1/\kap) \cdot  \sum_{i \sim j} x_{ij} \cdot \frac{\rho_i(\bmu, \x)}{1-\rho(\bmu,\x)}\\
&=(1/\kap) \cdot  \sum_{i \sim j} x_{ij} \cdot \bp{-1+\frac{1}{1-\rho_i(\bmu,\x)}}
\ge(1/\kap) \cdot  \sum_{i \sim j} x_{ij} \cdot \bp{-1+\frac{1}{1-\rho_i(\bomu,\x)}},
 \end{align*}
 where the inequality on the last line is due to
 \begin{align}\label{ineq:19-a}
\rho_i(\bmu,\x)=\sum_{\ell \sim i} x_{i \ell} \lam_\ell/\mu_{i\ell} \ge \sum_{\ell \sim i} x_{i \ell} \lam_\ell/\mu_{i}=\sum_{\ell \sim i} x_{i \ell} \lam_\ell/\bar{\mu}_{ij}=\rho_i(\bomu,\x).
\end{align}
Therefore, 
 \begin{align*}
\eta_t^*(\bmu) &=\max_{j \in J} \bw_j (\bmu,\x) \ge 
(1/\kap) \cdot \max_{j \in J} \sum_{i \sim j} x_{ij} \cdot \bp{-1+\frac{1}{1-\rho_i(\bomu,\x)}} \ge (1/\kap) \cdot \eta_t^*(\bomu),
\end{align*}
where the last inequality above is valid since $\x$ is optimal to $\pra(\bmu)$, and thus, it is feasible to $\pra(\bomu)$ since for each $j \in J$, $\sum_{i \sim j}x_{ij}=1$ and for each $i \in I$, $\rho_i(\bomu,\x) \le \rho_i(\bmu,\x) \le 1$ due to Inequality~\eqref{ineq:19-a}.

Now, we show the second one. Note that we assume by default that $\prb (\bmu)$ is feasible with $\eta^*_b(\bmu) \le 1$. So is $\prb(\bomu)$.  We claim that the optimal value of $\eta^*_b(\bmu)$ remains invariant after removing the constraints of $\rho_i \le 1$ for all $i \in I$. Let \prc be the program of $\prb$ after removing the constraints $\rho_i \le 1$ for all $i \in I$, and suppose $\eta_c (\bmu, \x)$ is the corresponding value of \prc with respect to $\bmu$ and $\x$. Observe that $\eta_s^*(\bmu)=\eta_c^{*}(\bmu)$ and $\eta_s^*(\bomu)=\eta_c^*(\bomu)$. Consider an optimal solution $\x=(x_{ij})$ to $\prb(\bomu)$. We can verify that it is surely feasible to $\prc(\bmu)$. Thus,
\begin{align}\label{ineq:19-b}
\eta^*_b(\bomu)&=\eta_s(\bomu,\x)=\eta_c(\bomu, \x) \ge \eta_c(\bmu, \x)/\kap  \ge \eta_c^*(\bmu)/\kap=\eta_s^*(\bmu)/\kap,
\end{align}
 where (a) the first inequality on~\eqref{ineq:19-b} follows from that 
 \[
 \eta_c(\bomu, \x)=\max_{i \in I}\sum_{\ell \sim i} x_{i\ell} \cdot \lam_\ell/\bar{\mu}_{i\ell} =\max_{i \in I} \sum_{\ell \sim i} x_{i\ell} \cdot \lam_\ell/\mu_{i} \ge \max_{i \in I} \sum_{\ell \sim i} x_{i\ell} \cdot \lam_\ell/(\kap \cdot \mu_{i \ell}) =\eta_c(\bmu, \x)/\kap;
 \]
 and (b) the second inequality on~\eqref{ineq:19-b} is valid since $\x$ is feasible to $\prc(\bmu)$ and $\eta_c^*(\bmu)$ is the optimal value.
\end{proof}
 
\begin{proof}[\tbf{Proof of Theorem~\ref{thm:main}}]
Let $\x=(x_{ij})$ be an optimal solution to $\prb(\bmu)$. Observe that $\x$ is feasible to $\pra(\bmu)$ since the two programs $\pra(\mu)$ and $\prb(\mu)$ share the same set of constraints.
\begingroup
\allowdisplaybreaks
\begin{align*}
\eta_t(\bmu, \x) &=\max_{j \in J} \sum_{i \sim j} x_{ij} \cdot \mu_{ij} \cdot \frac{\sum_{\ell \sim i} x_{i \ell} \lam_\ell/\mu_{i \ell}^2}{1-\sum_{\ell \sim i} x_{i \ell} \lam_\ell/\mu_{i \ell}}\\
&\le \max_{j \in J} \kap \cdot \sum_{i \sim j} x_{ij} \cdot \frac{\sum_{\ell \sim i} x_{i \ell} \lam_\ell/\mu_{i \ell}}{1-\sum_{\ell \sim i} x_{i \ell} \lam_\ell/\mu_{i \ell}}=\kap \cdot \max_{j \in J}  \sum_{i \sim j} x_{ij} \cdot \bp{-1+\frac{1}{1-\rho_i(\bmu, \x)}}\\
& \le\kap \cdot \bp{-1+\frac{1}{1-\max_{i \in I} \rho_i(\bmu,\x)}}=\kap \cdot \bp{-1+\frac{1}{1-\eta_s^*(\bmu)}}\\
&= \kap \cdot {\frac{\eta_s^*(\bmu)}{1- \eta_s^*(\bmu) }} = \kap \cdot {\frac{\eta_s^*(\bomu)}{1- \eta_s^*(\bomu) }} \cdot \bp{\frac{\eta_s^*(\bmu)}{1- \eta_s^*(\bmu) }}/\bp{\frac{\eta_s^*(\bomu)}{1- \eta_s^*(\bomu) }}\\
&=\kap \cdot \eta_t^*(\bomu) \cdot \frac{\eta_s^*(\bmu)}{\eta_s^*(\bomu)} \cdot \frac{1-\eta_s^*(\bomu)}{1-\eta_s^*(\bmu)} ~~\sbp{\mbox{by Claims~\ref{claim:main-3} and~\ref{claim:>=} on $\bomu$}}\\
 &\le \kap^3 \cdot \eta_t^*(\bmu) \cdot\frac{1-\eta_s^*(\bmu)/\kap}{1-\eta_s^*(\bmu)}=\kap^3 \cdot \eta_t^*(\bmu) \cdot\bp{1+\bp{1-\frac{1}{\kap}} \cdot \frac{\eta_s^*(\bmu)}{1-\eta_s^*(\bmu)}} ~~\sbp{\mbox{by Lemma~\ref{lem:19-a}}}.
\end{align*}
\endgroup
\end{proof} 

\section{Proof of Lemma~\ref{lem:>=}}

We split the proof of Lemma~\ref{lem:>=} into the following two claims. For any $S \subseteq J$ with $S \neq \emptyset$, let $\lam(S)=\sum_{j \in S} \lam_j=\lam \cdot |S|$, and 
$\mu(S)=\sum_{i \in \pa(S)} \mu_i$ with $\pa(S)=\{i: \exists j \in S, i \sim j\}$ being set of neighbors of $S$.\footnote{Note that $\pa(S)$ includes all possible neighbors of nodes in $S$, but nodes in $\pa(S)$ may have neighbors beyond $S$.}
\begin{claim}\label{claim:main-3}
$\eta_s^*=\max_{S \subseteq J, S \neq \emptyset} \lam(S)/\mu(S)$.
\end{claim}
\begin{claim}\label{claim:>=}
$\eta_t^* +1\ge 1/\sbp{1-\lam(S)/\mu(S)}$ for any $S \subseteq J$ with $S \neq \emptyset$.
\end{claim}
The two claims above together establish Lemma~\ref{lem:>=}. We present the proofs of the two Claims.

\begin{proof}[{\tbf{Proof of Claim~\ref{claim:main-3}}}] 
We first show $\eta_s^* \ge \lam(S)/\mu(S)$ for any $S \neq \emptyset, S \subseteq J$. Consider an optimal solution $\x^*=(x_{ij})$ for $\obar{\prb}$. Let $x_i=\sum_{j \sim i, j \in J} x_{ij}$.
\begin{align*}
&\eta_s^*=\max_{i \in I} \rho_i(\x^*) \ge \max_{i \in \pa(S)}\rho_i(\x^*) =\max_{i \in \pa(S)} (\lam x_i/\mu_i)\\
&\ge \bp{\sum_{i \in \pa(S)} \lam x_i}/\bp{\sum_{i \in \pa(S)} \mu_i} \ge \lam(S)/\mu(S),
\end{align*}
where the last inequality follows from 
\begin{align*}
\sum_{i \in \pa(S)} \lam x_i&=\lam \sum_{i \in \pa(S)} \sum_{ j \sim i, j \in J} x_{ij}\\
& \ge \lam  \sum_{ j \in S} \sum_{i \in I, i \sim j} x_{ij}=\lam \cdot |S|=\lam(S).
\end{align*} 
Now we show $\eta_s^* \le \max_{S \neq \emptyset, S\subseteq J}\lam(S)/\mu(S)$.  For an optimal solution  $\x^*=(x_{ij})$ of $\obar{\prb}$, let $I(\x^*)=\{i \in I: \rho_i(\x^*)=\eta_s^*\}$ be the set of node $i$ with saturated load under $\x^*$. Let $\y^*=(y_{ij})$ be an optimal solution of  $\obar{\prb}$ such that $I(\y^*)$ has the smallest size. Let $J(\y^*)=\{j \in J: \exists i \in I(\y^*) \mbox{~with~} i \sim j, y_{ij}>0 \}$ be set of non-zero neighbors of $I(\y^*)$ with respect to $\y^*$.

We claim that $\pa(J(\y^*))=I(\y^*)$. We show by contradiction as follows. Suppose there is  some $j \in J(\y^*)$ such that (1) there exists some $i \sim j, i \in I(\y^*)$ with $y_{ij}>0$ and $\rho_i(\y^*)=\eta_s^*$ and (2) there exists some $\bi \sim j$ with $\bi \notin I(\y^*)$, \ie  
$\rho_{\bi}(\y^*)<\eta_s^*$. Consider the following perturbation: $\bar{y}_{ij} \gets y_{ij} -\ep$ and $\bar{y}_{\bi,j} \gets y_{\bi,j}+\ep$ for an appropriate value of $\ep>0$, we could end up with another optimal solution that either has a strictly smaller size of $I(\bar{\y}^{*})$  if $|I(\y^*)|>1$ or a strictly better optimal value if $|I(\y^*)|=1$, which contradicts our assumption. Observe that $\eta_s^*=\rho_i(\y^*)=(\lam y_i/\mu_i)$ for every $i \in I(\y^*)$ with $y_i=\sum_{j \sim i} y_{ij}$. Set $S^*=J(\y^*)$ with $\pa(S^*)=I(\y^*)$, we have
\begin{align*}
\eta_s^*&=\bp{\sum_{i \in I(\y^*)} \lam y_i}/\bp{\sum_{i \in I(\y^*)} \mu_i}=
\bp{\sum_{i \in I(\y^*)} \lam \sum_{j \sim i}y_{ij}}/\mu(S^*)\\
&=\bp{\sum_{j \in J(\y^*)} \lam \sum_{i \sim j}y_{ij}}/\mu(S^*)=\lam(S^*)/\mu(S^*).
\end{align*}
Therefore, we get $\eta_s^* \le \max_{S \neq \emptyset, S\subseteq J}\lam(S)/\mu(S)$.
\end{proof}

\begin{proof}[{\tbf{Proof of Claim~\ref{claim:>=}}}] 

Consider any optimal solution $\x^*=(x_{ij})$ for $\obar{\pra}$. Recall that $\Lam=\lam(S)=\lam \cdot |S|$ and $\Phi=\mu(S)=\sum_{i \in \pa(S)} \mu_i$. We see $\Lam/\Phi=\lam(S)/\mu(S) \le \eta_s^* \le 1$ due to Claim~\ref{claim:main-3}. Set $x_i= \sum_{j \sim i, j \in J} x_{ij}$ and $x_i(S)=\sum_{j \sim i, j \in S}x_{ij}$. Let $\rho_i(\x^*)=(\lam/\mu_i) \sum_{j \sim i, j\in J}x_{ij}=(\lam/\mu_i) \cdot x_i$ and  $\rho_{i,S}(\x^*)=(\lam/\mu_i)  \sum_{j \sim i, j\in S} x_{ij}=(\lam/\mu_i) \cdot x_i(S)$. Thus, $\rho_{i}(\x^*) \ge \rho_{i,S}(\x^*)$.
\begingroup
\allowdisplaybreaks
\begin{align*}
&\eta_t^*+1 =\max_{j \in J} \sbp{\bw_j(\x_a^*) +1}\ge \frac{1}{|S|} \sum_{j \in S} \sbp{\bw_j(\x_a^*) +1} = \frac{1}{|S|} \sum_{j \in S} \sum_{i \sim j} \frac{x_{ij}}{1-\rho_i(\x^*)}
= \frac{1}{|S|} \sum_{i \in \pa(S)} \sum_{j \sim i, j \in S} \frac{x_{ij}}{1-\rho_i(\x^*)}\\
&=\frac{1}{|S|} \sum_{i \in \pa(S)}  \frac{x_i(S)}{1-\rho_i(\x^*)} \ge 
\frac{1}{|S|} \sum_{i \in \pa(S)}  \frac{x_i(S)}{1-\rho_{i,S}(\x^*)} \\
&=\frac{1}{\lam \cdot |S|} \sum_{i \in \pa(S)}  \mu_i \cdot \frac{(\lam/\mu_i) \cdot x_i(S)}{1-\rho_{i,S}(\x^*)}=\frac{1}{\Lam}\sum_{i \in \pa(S)}   \frac{\mu_i \cdot \rho_{i,S}(\x^*)}{1-\rho_{i,S}(\x^*)}=\frac{1}{\Lam}\sum_{i \in \pa(S)}  \mu_i \cdot \bp{-1+\frac{1}{1-\rho_{i,S} (\x^*)}}\\
&=-\frac{\Phi}{\Lam}+\frac{1}{\Lam}\sum_{i \in \pa(S)} \frac{\mu_i}{1-\rho_{i,S} (\x^*)}=
-\frac{\Phi}{\Lam}+\frac{1}{\Lam}\sum_{i \in \pa(S)} \frac{\mu_i}{1-(\lam/\mu_i) \cdot x_i(S)}=
-\frac{\Phi}{\Lam}+\frac{1}{\Lam}\sum_{i \in \pa(S)} \frac{\mu_i^2}{\mu_i-\lam \cdot x_i(S)}.
\end{align*}
\endgroup
Set $z_i=\lam \cdot x_i(S)$. Observe that (1)
\begin{align*}
\sum_{i \in \pa(S)}z_i &=\sum_{i \in \pa(S)} \lam \cdot x_i(S) =\lam\sum_{i \in \pa(S)} \sum_{j \sim i, j \in S} x_{ij}=\lam \cdot \sum_{j \in S} \sum_{i \sim j, i \in I} x_{ij}=\lam \cdot |S|=\Lam,
\end{align*}
and (2) for any $i \in I$, we have $z_i \le \mu_i$ since $z_i/\mu_i=(\lam/\mu_i) \cdot x_i(S)=\rho_{i,S} (\x^*) \le \rho_i(\x^*)\le 1$. For any fixed values of $(\mu_i)$, let $g_i(z_i):=\mu_i^2/(\mu_i-z_i)$ be a function of $z_i \in [0, \mu_i]$. Consider a minimization program below,
\begin{align}\label{eqn:claim-1}
\bp{ \min \sum_{i \in \pa(S)} g_i(z_i): 0 \le z_i \le \mu_i, \forall i\in \pa(S); \sum_{i \in \pa(S)}z_i=\Lam.}
\end{align}
By local perturbation, we claim that for any given $(\mu_i)$, Program~\eqref{eqn:claim-1} has a unique optimal solution $\z^*=(z_i^*)$ such that $\mu_i/(\mu_i-z^*_i)$ takes a uniform value for every $i \in \pa(S)$, and so does $z^*_i/\mu_i$. Let $\beta=z^*_i/\mu_i$ for every $i \in \pa(S)$. We see that $\beta=\sum_{i \in \pa(S)}z^*_i/ \sum_{i \in \pa(S)} \mu_i=\Lam/\Phi$. Thus, we claim that
\begin{align*}
\eta_t^*+1 &\ge  -\frac{\Phi}{\Lam}+\frac{1}{\Lam}\sum_{i \in \pa(S)} \frac{\mu_i^2}{\mu_i-\lam \cdot x_i(S)}= -\frac{\Phi}{\Lam}+\frac{1}{\Lam}\sum_{i \in \pa(S)} g_i(z_i)\\
&\ge 
 -\frac{\Phi}{\Lam}+\frac{1}{\Lam}\sum_{i \in \pa(S)} g_i(z_i^*)= -\frac{\Phi}{\Lam}+\frac{1}{\Lam}\sum_{i \in \pa(S)} \frac{\mu_i}{1-z_i^*/\mu_i}\\
 &=-\frac{\Phi}{\Lam}+\frac{\Phi}{\Lam} \cdot  \frac{1}{1-\Lam/\Phi}=\frac{\Phi}{\Lam} \cdot \bp{\frac{1}{1-\Lam/\Phi}-1}=\frac{1}{1-\Lam/\Phi}=\frac{1}{1-\lam(S)/\mu(S)}.
\end{align*}
\end{proof}

\section{Details for the Experiments Settings}
\subsection{Approaches Used for Determining the Duration of the Tasks}
\label{app:durapp}
We propose three approaches for determining the duration of the tasks.
In the first approach, we aim to study settings where $\kappa=1$. We calculated the average duration time for each (human) teleoperator across all allowed tasks and computed the overall average. This average was then used as the mean in the exponential distribution. By setting $\kappa=1$ in this method, the optimal solution of $\prb$ also became an optimal solution of $\pra$.
In the second approach, we aim at varying $\kappa$ by assigning different mean values to the various tasks. These mean values were chosen around the mean calculated in the previous approach. This allowed us to assess the performance of our algorithms and heuristics in scenarios where there was significant variance in the duration times of the different tasks.
Finally, the third approach considers the average time taken by the teleoperators to perform the tasks in the simulation for each specific teleoperator-task combination. This approach helps us to make the settings as similar as possible to the real data.

The exact implementation of the second approach is as follows:
The mean duration for each teleoperator is the average task duration she is authorized to undertake. The durations follow an exponential distribution with means calculated as $(2\cdot \kappa\cdot x) / (1 + k)$ and $2x - (2\cdot \kappa\cdot x) / (1 + k)$ for the first and last tasks respectively, where $x$ represents the average task duration for the teleoperator. The average values range from 3.33 to 8 seconds, with a standard deviation 1.58. For the remaining tasks, the means are uniformly selected within the range between the first and last task means. In cases where a teleoperator is permitted to perform only one task, the mean is set as $x$ for the exponential distribution.
The experiments in the main paper use the second approach, while both the second and third approaches are presented in the appendix.

\subsection{Generating the Task Arrival Rates}
To establish the task arrival rates, we followed a specific procedure. We used the average number of task arrivals per day (100,000) used in Viden et al.~\cite{viden} as a basis and examined the neighborhood of this number (e.g., 60000-140000). For each task type, we multiplied this average by the weight assigned to that specific task type, resulting in the final $\lambda$ parameter for the normal distribution of task types. 

Having the task arrival rates, the actual arrival times are generated by the following procedure:
Using an exponential distribution with a mean of $1/\lambda$ we calculate a set of arrival times. For each arrival time we decide the type of the arriving task by using configurable probabilities.
This approach allowed us to model the arrival times of tasks and analyze the system's performance under various task-type balances.

\subsection{Experiments Environment and Some More Technical Details}
In our experiments, we ran simulations using Python and Matlab. The simulation ran for a (virtual) period of 4 weeks. The nonlinear minimax problem $\pra$ was solved with the Matlab function fmincon and the problem $\prb$ was cast as a linear program and solved with the Matlab function linprog. Most of the experiments were carried out on a Windows laptop. Other experiments were run on a Linux server with 98 cores to save time.
It is important to emphasize that the \prb method is valid only when the available workers can handle all the tasks. Therefore, we focus on problems for which $\prb$ can provide a feasible solution.

\section{Additional Figures for the Experiments Section}
\subsection{Results for a Different Duration Distribution}

Figures \ref{fig:load2} and \ref{fig:lam2} are similar to figures ~\ref{fig:window}(c,d,e,f) from the Experiments section. The only difference is that the parameter $\mu_{i,j}$ is defined as the real average time it takes the teleoperator $i$ to perform a task of type $j$. Therefore, the value of $\kappa$ varies for the different workers.

\begin{figure}
    \centering
    \includegraphics[width=0.33\columnwidth]{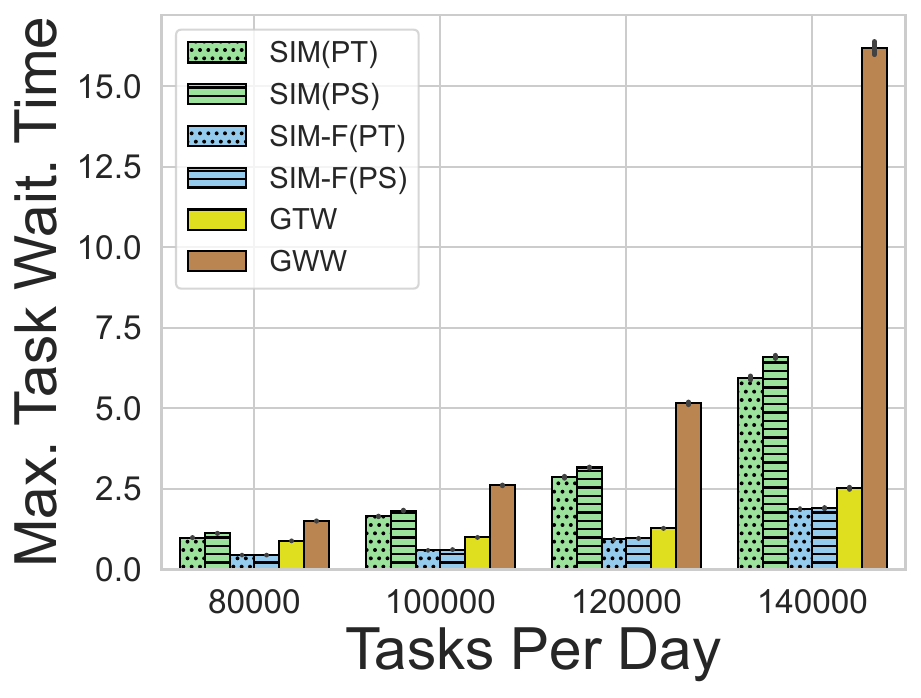}
    \includegraphics[width=0.33\columnwidth]{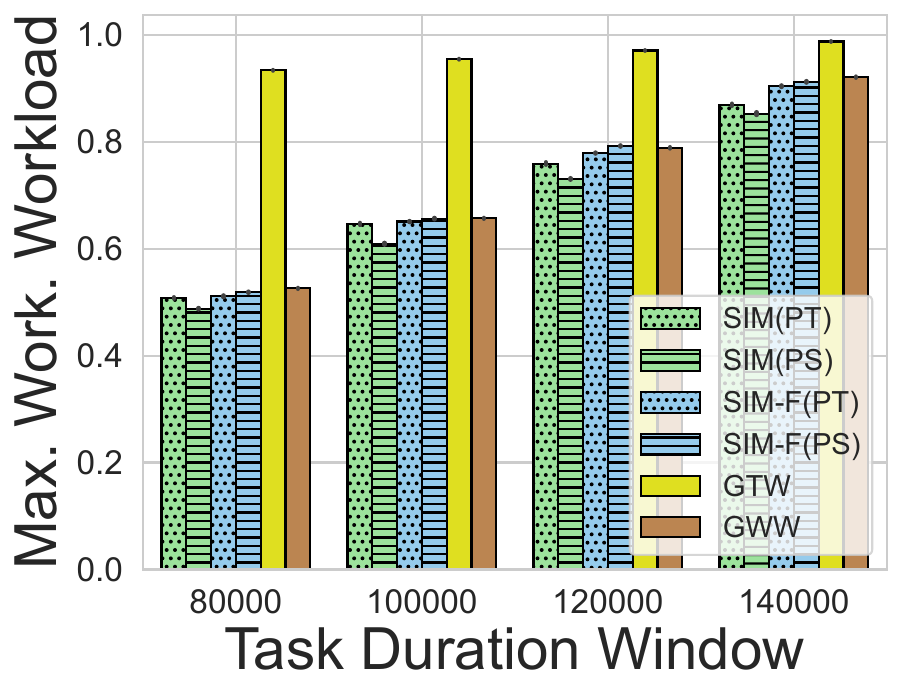}
    \caption{Maximum Task Waiting Time (left) and Maximum worker workload (right) under varying task arrival loads. The x-axis represents the expected number of arrivals per day. In these graphs, the arrival distributions of all tasks are equal.} 
    \label{fig:load2}
\end{figure}

\begin{figure}
    \centering
    \includegraphics[width=0.33\columnwidth]{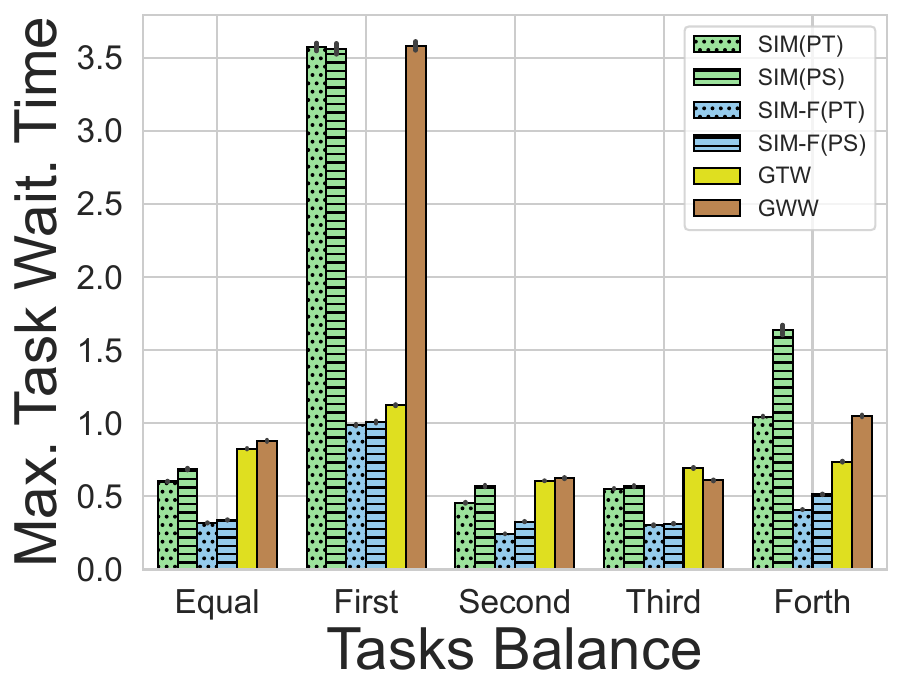}
    \includegraphics[width=0.33\columnwidth]{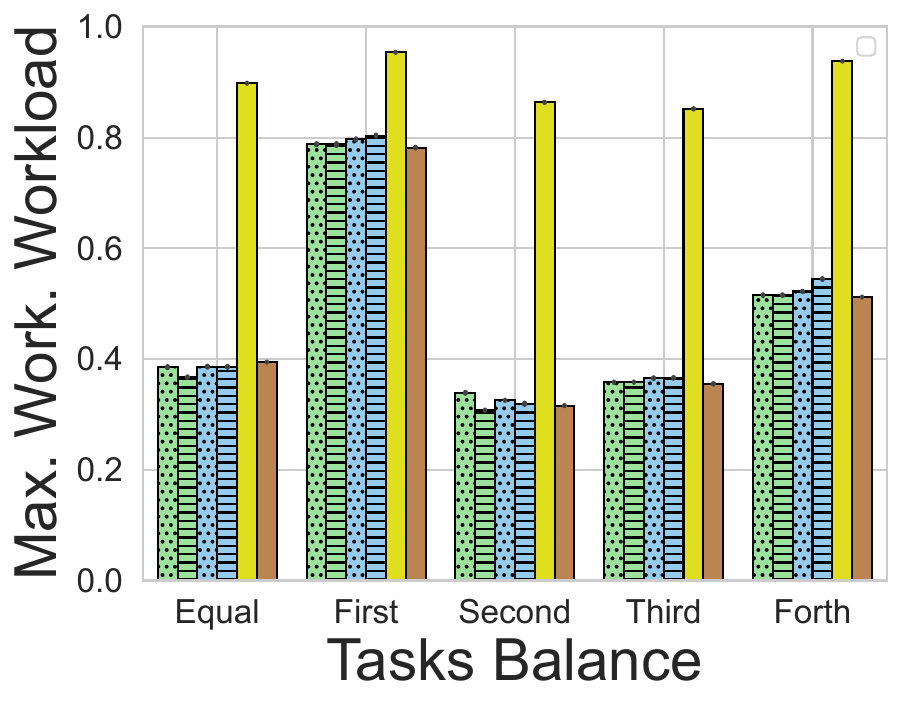}
    \caption{Task waiting time (left) and worker workload (right) for different arrival distributions: equal task arrival vs. one task type arriving seven times more than the others (task load was 60000 requests per day). }
    \label{fig:lam2}
\end{figure}

\subsection{Results for Different Numbers of Task Types}

In this section, we vary the number of task types and analyze the impact on the maximum waiting time of tasks and the maximum workload of workers. In Figure \ref{fig:apnnewtasks} we used the second approach(see section \ref{app:durapp}) to define the task duration parameter ($\mu_{i,j}$) for tasks 1-4, while in Figure~\ref{fig:apnorigtasks} we used the third approach.
In both Figures, to define the parameters for a new task not present in the dataset (5 and 6), we selected a mean value uniformly at random from the range 4 and 11 (the means of the original tasks) and used the standard deviation of the actual duration values to define a normal distribution. This normal distribution was used to determine the duration for each worker. To obtain a reasonable duration, we set a minimum duration of at least 1 second for all workers.

The graph of the input network was created according to these durations. An edge exists between a worker and a task type only if the average time required by the worker to complete a task of that type is at most equal to the median duration for that task type. In the final step, the duration parameters for the new tasks were determined based on these durations, using the second and third approaches for Figures~\ref{fig:apnnewtasks} and \ref{fig:apnorigtasks} respectively.

In Figures \ref{fig:apnorigtasks} and \ref{fig:apnnewtasks} we see that having only 2 task types leads to a worse performance both in the maximum waiting time of the tasks and in the maximum worker workload measures. For 3-6 task types, we find that the performance in Figure \ref{fig:apnorigtasks} decreases for both measures as the number of task types increases (at least for some methods), while the performance in Figure \ref{fig:apnnewtasks} is very similar in this range.

We conclude that although the number of task types is irrelevant when distributed uniformly among the workers, it is quite important if they are distributed differently (as in the real-data). This fact should be taken into account when modeling task types in an application.

\begin{figure}
    \centering
    \includegraphics[width=0.33\columnwidth]{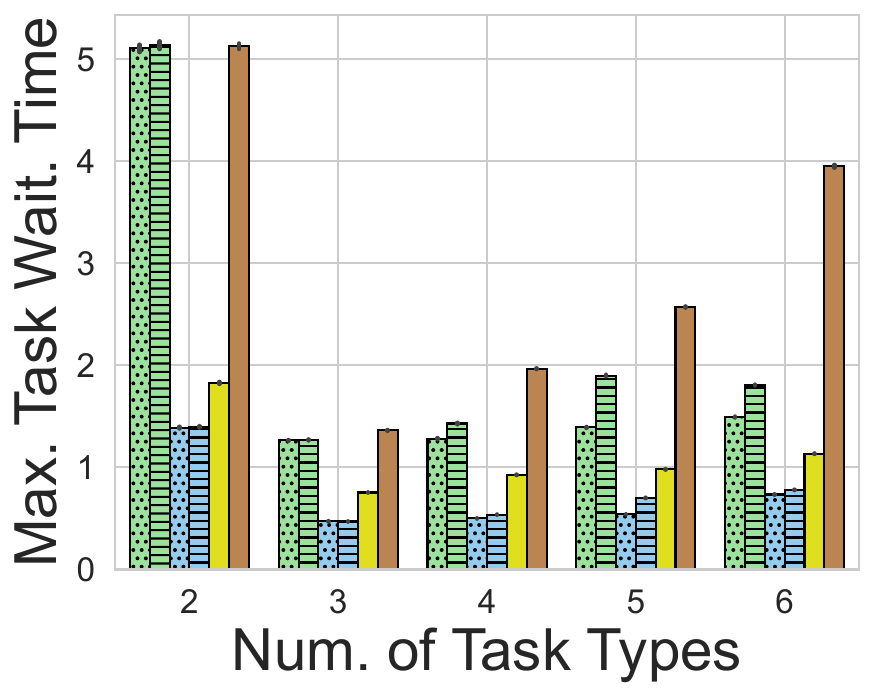}
    \includegraphics[width=0.33\columnwidth]{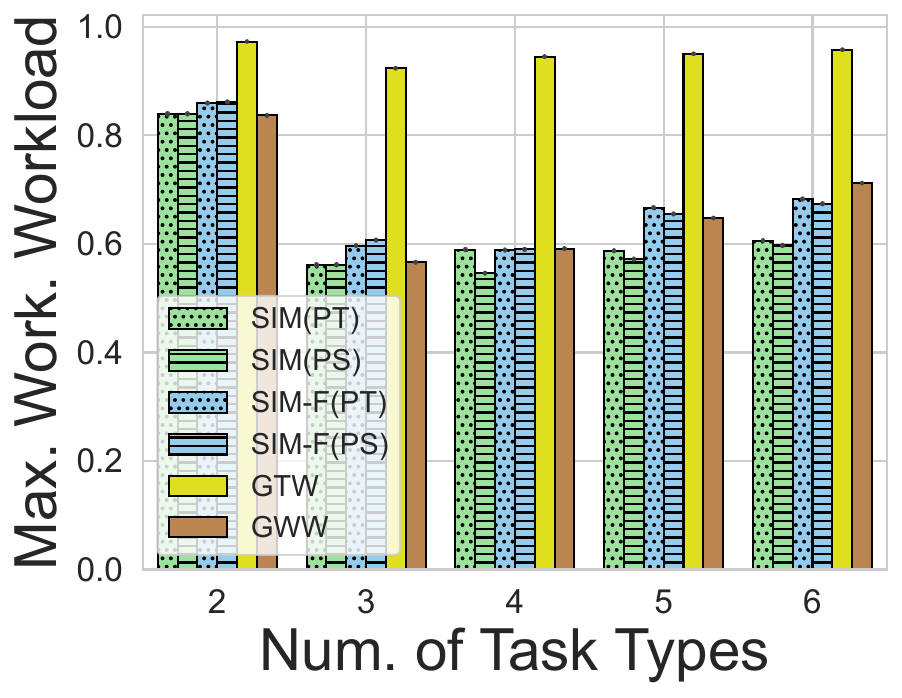}
    \caption{Task waiting time (left) and worker workload (right) for different numbers of task types(task load was 90000 requests per day; $\mu_{i,j}$ varies according to the real and synthetic data and $\kappa$ is concluded accordingly). }
    \label{fig:apnorigtasks}
\end{figure}

\begin{figure}
    \centering
    \includegraphics[width=0.33\columnwidth]{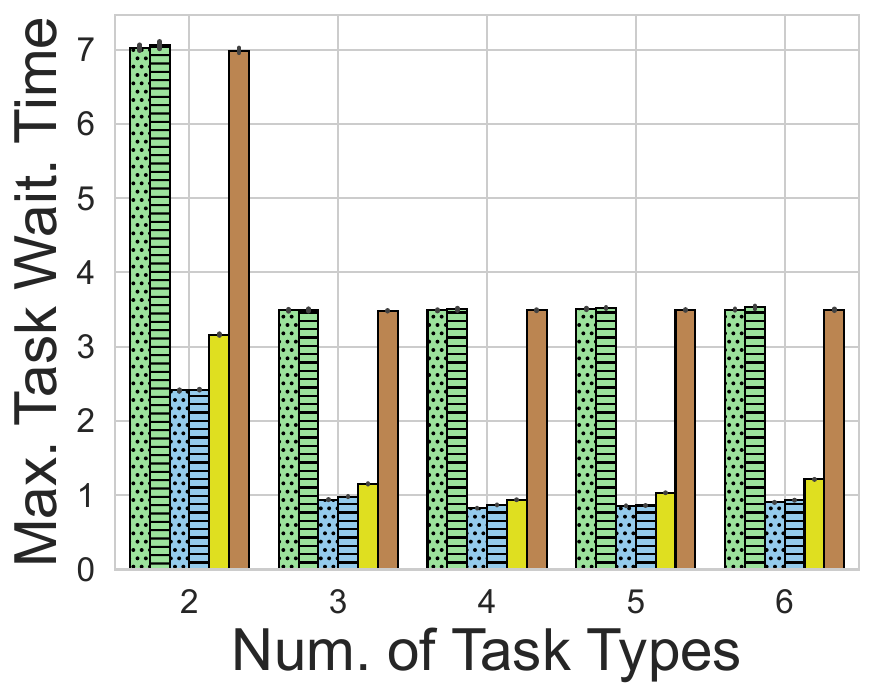}
    \includegraphics[width=0.33\columnwidth]{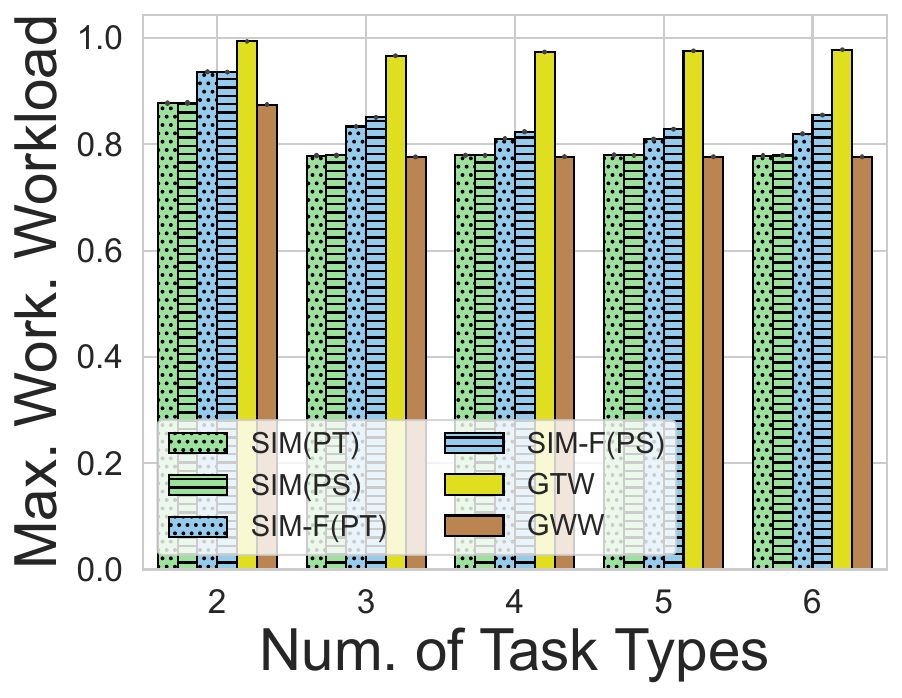}
    \caption{Task waiting time (left) and worker workload (right) for different numbers of task types(task load was 120000 requests per day and $\kappa=1$). }
    \label{fig:apnnewtasks}
\end{figure}

\subsection{Results for Different Numbers of Workers}

In this section, we vary the number of workers and analyze the effects on the maximum waiting time of the tasks and the maximum workload of the workers. In Figure \ref{fig:apnnewworkers} we used the second approach(section \ref{app:durapp}) to define the task duration parameter ($\mu_{i,j}$) for workers 1-9, while in Figure~\ref{fig:apnorigworkers} we used the third approach.
In both Figures, we have chosen a (uniformly) random duration from the range of existing durations to define the average durations for a new worker (10 to 14).
The graph of the input network was defined using these durations, where an edge between a worker and a task type exists only if the average duration of the worker performing a task of that type is at most equal to the median duration for that task type. Finally, the duration parameters of the new tasks were determined based on these durations according to the second and third approaches.

We find that, as expected, the performance of the system improves as the number of workers increases. Therefore, when deciding how many workers to acquire, we should look for a reasonable balance between system performance and the cost of additional workers.

\begin{figure}
    \centering
    \includegraphics[width=0.33\columnwidth]{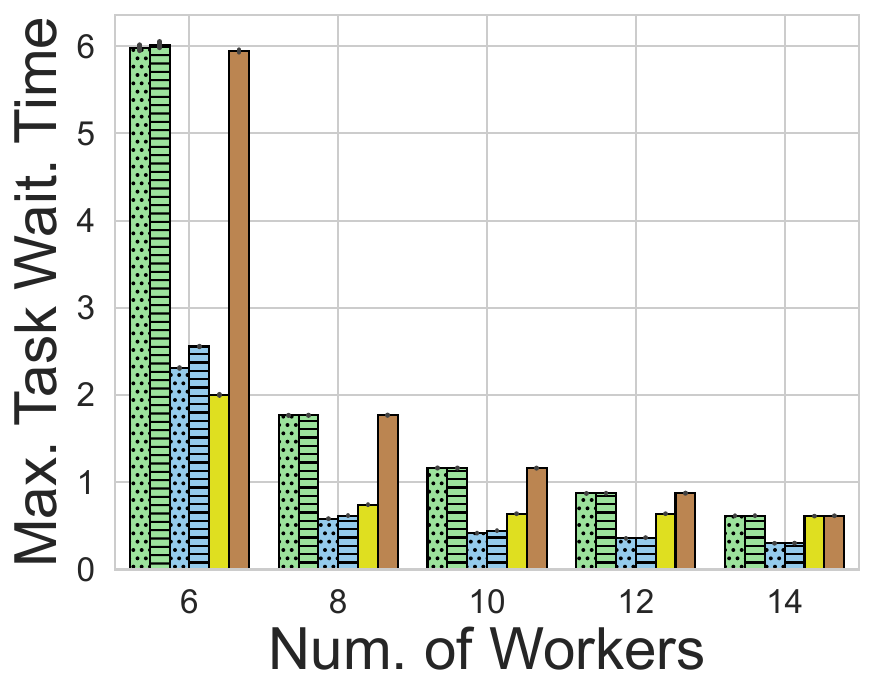}
    \includegraphics[width=0.33\columnwidth]{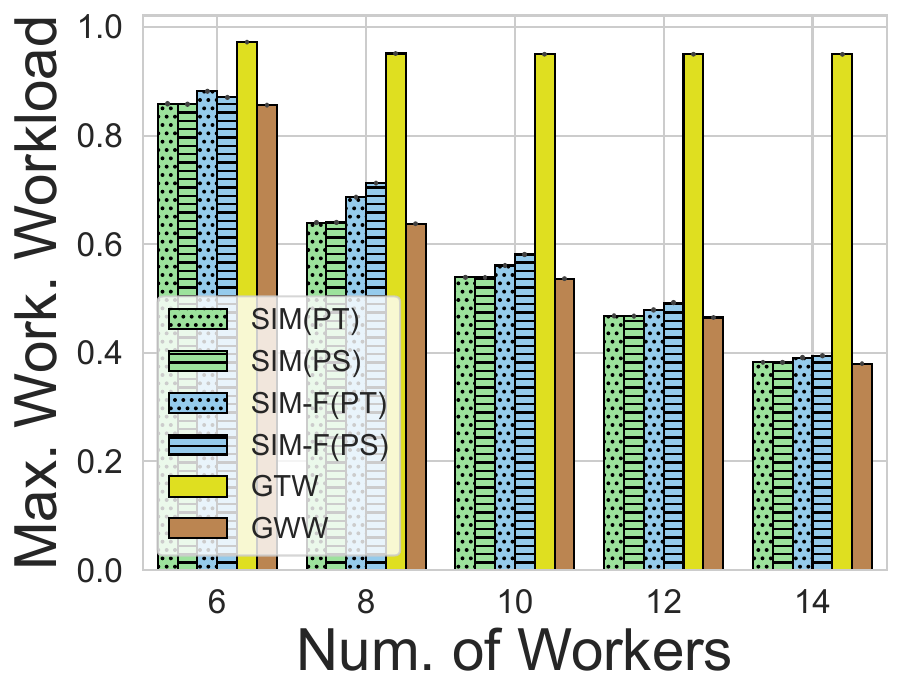}
    \caption{Task waiting time (left) and worker workload (right) for different numbers of workers(task load was 90000 requests per day and $\kappa$ is determined by the values of $\mu_{i,j}$). }
    \label{fig:apnnewworkers}
\end{figure}
\begin{figure}
    \centering
    \includegraphics[width=0.33\columnwidth]{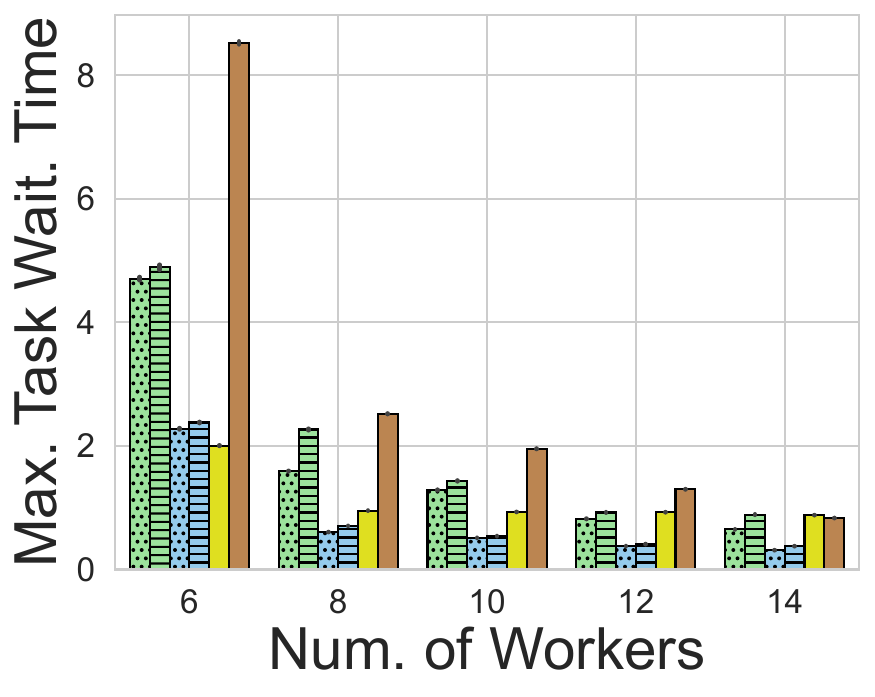}
    \includegraphics[width=0.33\columnwidth]{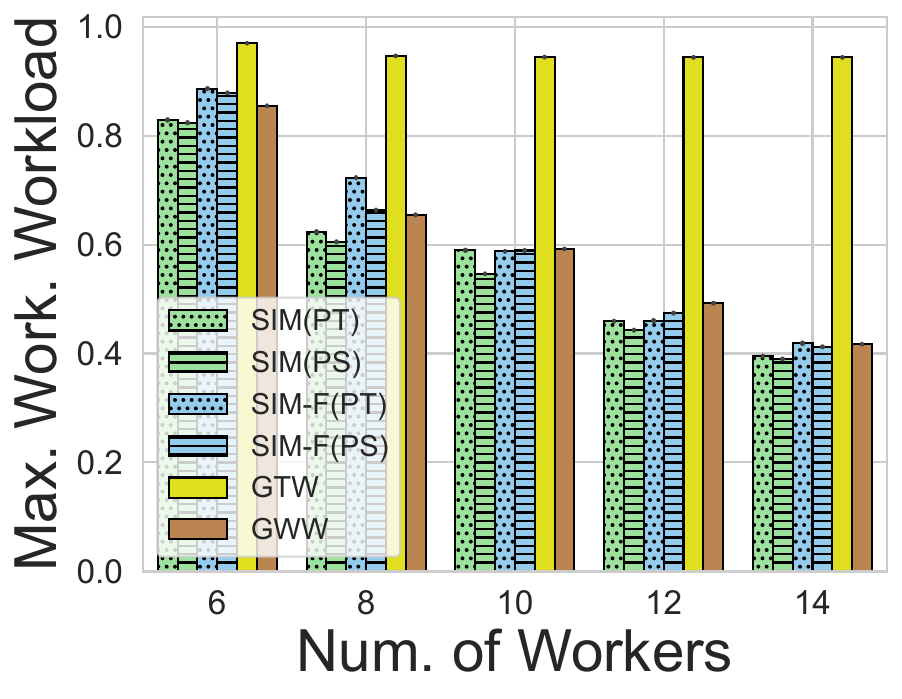}
    \caption{Task waiting time (left) and worker workload (right) for different numbers of workers types(task load was 90000 requests per day and $\kappa=1$ ). }
    \label{fig:apnorigworkers}
\end{figure}

\fi


\end{document}